\documentclass[12pt]{elsarticle}
\usepackage[margin=2.5cm]{geometry}
\usepackage{lipsum}

\makeatletter
\def\ps@pprintTitle{%
   \let\@oddhead\@empty
   \let\@evenhead\@empty
   \def\@oddfoot{\reset@font\hfil\thepage\hfil}
   \let\@evenfoot\@oddfoot
}
\makeatother

\usepackage{amsfonts,color,morefloats,pslatex}
\usepackage{amssymb,amsthm, amsmath,latexsym}

\newtheorem{theorem}{Theorem}
\newtheorem{lemma}[theorem]{Lemma}

\newtheorem{corollary}[theorem]{Corollary}

\newtheorem{open}[theorem]{Open Problem}

\newtheorem{example}[theorem]{Example}

\newcommand{\ord}{{\mathrm{ord}}}

\newcommand{\rank}{{\mathrm{rank}}}
\newcommand{\lcm}{{\mathrm{lcm}}}
\newcommand{\tr}{{\mathrm{Tr}}}

\newcommand{\gf}{{\mathrm{GF}}}

\newcommand{\support}{{\mathrm{suppt}}}

\newcommand{\PAut}{{\mathrm{PAut}}}
\newcommand{\MAut}{{\mathrm{MAut}}}
\newcommand{\GAut}{{\mathrm{Aut}}}

\newcommand{\Sym}{{\mathrm{Sym}}}

\newcommand{\wt}{{\mathtt{wt}}}

\newcommand{\Z}{\mathbb{{Z}}}

\newcommand{\m}{\mathbb{M}}

\newcommand{\cP}{{\mathcal{P}}}
\newcommand{\cB}{{\mathcal{B}}}

\newcommand{\C}{{\mathsf{C}}}

\newcommand{\bc}{{\mathbf{c}}}

\newcommand{\bzero}{{\mathbf{0}}}

\newcommand{\bD}{{\mathbb{D}}}

\newcommand{\PGL}{{\mathrm{PGL}}}

\newcommand{\PGaL}{{\mathrm{P \Gamma L}}}

\usepackage{blindtext}

\begin{document}

\begin{frontmatter}




\title{An infinite family of linear codes supporting  $4$-designs}

\tnotetext[fn1]{
C. Tang was supported by The National Natural Science Foundation of China (Grant No.
11871058) and China West Normal University (14E013, CXTD2014-4 and the Meritocracy Research
Funds).
C. Ding's research was supported by the Hong Kong Research Grants Council,
Proj. No. 16300418. 
}

\author[cmt]{Chunming Tang}
\ead{tangchunmingmath@163.com} 

\author[cding]{Cunsheng Ding}
\ead{cding@ust.hk}

\address[cmt]{School of Mathematics and Information, China West Normal University, Nanchong, Sichuan,  637002, China}
\address[cding]{Department of Computer Science and Engineering, The Hong Kong University of Science and Technology, Clear Water Bay, Kowloon, Hong Kong, China} 




\begin{abstract} 
The first linear code supporting a $4$-design was the $[11, 6, 5]$ ternary Golay code 
discovered in 1949 by Golay. In the past 71 years, sporadic linear codes 
holding $4$-designs or $5$-designs were discovered and many infinite families of linear codes supporting $3$-designs were constructed. 
However, the question as to whether there is an infinite family of linear codes holding an infinite 
family of $t$-designs for $t\geq 4$ remains open for 71 years. This paper settles this long-standing problem by 
presenting an infinite family of BCH codes of length $2^{2m+1}+1$ over $\gf(2^{2m+1})$ holding an infinite family of $4$-$(2^{2m+1}+1, 6, 2^{2m}-4)$ designs. Moreover,
 an infinite family of linear codes holding the spherical design $S(3, 5, 4^m+1)$ is presented. 
      
\end{abstract}

\begin{keyword}
 BCH code \sep Cyclic code  \sep linear code \sep $t$-design, elementary symmetric polynomial.

\MSC  05B05 \sep 51E10 \sep 94B15

\end{keyword}

\end{frontmatter}


\section{Introduction}

Let $\cP$ be a set of $v \ge 1$ elements, and let $\cB$ be a set of $k$-subsets of $\cP$, where $k$ is
a positive integer with $1 \leq k \leq v$. Let $t$ be a positive integer with $t \leq k$. The incidence structure 
$\bD = (\cP, \cB)$ is called a $t$-$(v, k, \lambda)$ {\em design\index{design}}, or simply {\em $t$-design\index{$t$-design}}, if every $t$-subset of $\cP$ is contained in exactly $\lambda$ elements of
$\cB$. The elements of $\cP$ are called points, and those of $\cB$ are referred to as blocks. 
The set $\cB$ is called the block set. 
We usually use $b$ to denote the number of blocks in $\cB$. 
Let $\binom{\cP}{k}$ denote the set of all $k$-subsets of $\cP$. Then $\left (\cP, \binom{\cP}{k} \right )$ is a $k$-$(v, k, 1)$ design, 
which is called a \emph{complete design}. 
  A $t$-design is called {\em simple\index{simple}} if $\cB$ does not contain
any repeated blocks.
In this paper, we consider only simple $t$-designs with $v > k > t$.
A $t$-$(v,k,\lambda)$ design is referred to as a
{\em Steiner system\index{Steiner system}} if $t \geq 2$ and $\lambda=1$,
and is denoted by $S(t,k, v)$. The parameters of a $t$-$(\nu,k, \lambda )$ design satisfy:
\begin{align*}
\binom{\nu}{t} \lambda =\binom{k}{t} b.
\end{align*}

A $t$-$(v, k, \lambda)$ design is also a $s$-$(v, k, \lambda_s)$ design with 
\begin{eqnarray}\label{eqn-lambdas}
\lambda_s=\lambda \binom{v-s}{t-s}/\binom{k-s}{t-s}
\end{eqnarray} 
for all $s$ with $0 \leq s \leq t$.

Let $\C$ be a $[v, \kappa, d]$ linear code over $\gf(q)$. Let $A_i$ denote the
number of codewords with Hamming weight $i$ in $\C$, where $0 \leq i \leq v$. The sequence
$(A_0, A_1, \cdots, A_{v})$ is
called the \textit{weight distribution} of $\C$, and $\sum_{i=0}^v A_iz^i$ is referred to as
the \textit{weight enumerator} of $\C$. In this paper, $\C^\perp$ denotes the dual code 
of $\C$, $d^\perp$ denotes the minimum distance of $\C^\perp$, and $(A_0^\perp, A_1^\perp, \cdots, A_{v}^\perp)$ 
denotes the weight distribution of $\C^\perp$.  

A  $[v, \kappa, d]$ linear code over $\gf(q)$ is said to be distance-optimal if there is no $[v, \kappa, d']$ over 
$\gf(q)$ with $d'>d$. A  $[v, \kappa, d]$ linear code over $\gf(q)$ is said to be dimension-optimal if there is no 
$[v, \kappa', d]$ over $\gf(q)$ with $\kappa' > \kappa$.  A  $[v, \kappa, d]$ linear code over $\gf(q)$ is said 
to be length-optimal if there is no $[v', \kappa, d]$ over $\gf(q)$ with $v'<v$. A linear code is said to be optimal 
if it is distance-optimal, dimension-optimal and length-optimal.  

A coding-theoretic construction of $t$-designs is the following. 
For each $k$ with $A_k \neq 0$,  let $\cB_k(\C)$ denote
the set of the supports of all codewords with Hamming weight $k$ in $\C$, where the coordinates of a codeword
are indexed by $(p_1, \ldots, p_v)$. Let $\cP(\C)=\{p_1, \ldots, p_v\}$.  The pair $(\cP, \cB_k(\C))$
may be a $t$-$(v, k, \lambda)$ design for some positive integer $\lambda$, which is called a
\emph{support design} of the code, and is denoted by $\bD_k(\C)$. In such a case, we say that the code $\C$ holds a $t$-$(v, k, \lambda)$
design or the codewords of weight $k$ in $\C$ support a $t$-$(v, k, \lambda)$
design. 

The following theorem, developed by Assumus and Mattson, shows that the pair $(\cP(\C), \cB_k(\C))$ defined by 
a linear code is a $t$-design under certain conditions \cite{AM69}.

\begin{theorem}[Assmus-Mattson Theorem]\label{thm-designAMtheorem}
Let $\C$ be a $[v,k,d]$ code over $\gf(q)$. Let $d^\perp$ denote the minimum distance of $\C^\perp$. 
Let $w$ be the largest integer satisfying $w \leq v$ and 
$$ 
w-\left\lfloor  \frac{w+q-2}{q-1} \right\rfloor <d. 
$$ 
Define $w^\perp$ analogously using $d^\perp$. Let $(A_i)_{i=0}^v$ and $(A_i^\perp)_{i=0}^v$ denote 
the weight distribution of $\C$ and $\C^\perp$, respectively. Fix a positive integer $t$ with $t<d$, and 
let $s$ be the number of $i$ with $A_i^\perp \neq 0$ for $1 \leq i \leq v-t$. Suppose $s \leq d-t$. Then 
\begin{itemize}
\item the codewords of weight $i$ in $\C$ hold a $t$-design provided $A_i \neq 0$ and $d \leq i \leq w$, and 
\item the codewords of weight $i$ in $\C^\perp$ hold a $t$-design provided $A_i^\perp \neq 0$ and 
         $d^\perp \leq i \leq \min\{v-t, w^\perp\}$. 
\end{itemize}
\end{theorem}

The Assmus-Mattson Theorem is a very useful tool in constructing $t$-designs from linear codes 
(see, for example, \cite{DingLi16}, \cite{Dingbook18}). 
 A generalized Assmus-Mattson theorem is developed in \cite{TDX19}.
Another sufficient condition for the 
incidence structure $(\cP, \cB_k)$ to be a $t$-design is via the automorphism group of the 
code $\C$.

The set of coordinate permutations that map a code $\C$ to itself forms a group, which is referred to as
the \emph{permutation automorphism group\index{permutation automorphism group of codes}} of $\C$
and denoted by $\PAut(\C)$. If $\C$ is a code of length $n$, then $\PAut(\C)$ is a subgroup of the
\emph{symmetric group\index{symmetric group}} $\Sym_n$.

A \emph{monomial matrix\index{monomial matrix}} over $\gf(q)$ is a square matrix having exactly one
nonzero element of $\gf(q)$  in each row and column. A monomial matrix $M$ can be written either in
the form $DP$ or the form $PD_1$, where $D$ and $D_1$ are diagonal matrices and $P$ is a permutation
matrix.

The set of monomial matrices that map $\C$ to itself forms the group $\MAut(\C)$,  which is called the
\emph{monomial automorphism group\index{monomial automorphism group}} of $\C$. Clearly, we have
$$
\PAut(\C) \subseteq \MAut(\C).
$$

The \textit{automorphism group}\index{automorphism group} of $\C$, denoted by $\GAut(\C)$, is the set
of maps of the form $M\gamma$,
where $M$ is a monomial matrix and $\gamma$ is a field automorphism, that map $\C$ to itself. In the binary
case, $\PAut(\C)$,  $\MAut(\C)$ and $\GAut(\C)$ are the same. If $q$ is a prime, $\MAut(\C)$ and
$\GAut(\C)$ are identical. In general, we have
$$
\PAut(\C) \subseteq \MAut(\C) \subseteq \GAut(\C).
$$

By definition, every element in $\GAut(\C)$ is of the form $DP\gamma$, where $D$ is a diagonal matrix,
$P$ is a permutation matrix, and $\gamma$ is an automorphism of $\gf(q)$.
The automorphism group $\GAut(\C)$ is said to be $t$-transitive if for every pair of $t$-element ordered
sets of coordinates, there is an element $DP\gamma$ of the automorphism group $\GAut(\C)$ such that its
permutation part $P$ sends the first set to the second set. The automorphism group $\GAut(\C)$ is said to be $t$-homogeneous if for every pair of $t$-element 
sets of coordinates, there is an element $DP\gamma$ of the automorphism group $\GAut(\C)$ such that its
permutation part $P$ sends the first set to the second set.

The next theorem gives a 
sufficient condition for a linear code to hold $t$-designs \cite[p. 308]{HP03}.

\begin{theorem}\label{thm-designCodeAutm}
Let $\C$ be a linear code of length $n$ over $\gf(q)$ where $\GAut(\C)$ is $t$-transitive 
or $t$-homogeneous. Then the codewords of any weight $i \geq t$ of $\C$ hold a $t$-design.
\end{theorem} 

So far many infinite families of $t$-designs with $t=2, 3$ have been constructed from this coding-theoretic 
approach. However, no infinite family of $4$-designs has been produced with this approach, 
though sporadic $t$-designs for $t=4,5$ have been obtained from linear codes. 
The first linear code supporting $t$-design with $t\ge 4$ was the $[11, 6, 5]$ ternary Golay code 
discovered in 1949 by Golay \cite{Golay49}. This ternary code holds $4$-designs, and its extended code holds a Steiner system 
$S(5, 6, 12)$ having the largest strength known. In the past 71 years, sporadic linear codes 
holding $4$-designs or $5$-designs were discovered and many infinite families of linear codes supporting $3$-designs were 
constructed. However, the question as to whether there is an infinite family of liner codes holding an infinite 
family of $t$-designs for $t\geq 4$ remains open for 71 years. 
This paper settles this long-standing problem by presenting an infinite family of near MDS codes over $\gf(2^{2m+1})$ holding an infinite family of $4$-$(2^{2m+1}+1, 6, 2^{2m}-4)$ designs. 
In addition, we present an infinite family of linear codes holding the spherical design 
$S(3, 5, 1+4^m)$.

\section{Cyclic codes, BCH codes, AMDS codes and NMDS codes} 

\subsection{Cyclic codes and BCH codes}

An $[n,k, d]$ code $\C$ over $\gf(q)$ is {\em cyclic} if 
$(c_0,c_1, \cdots, c_{n-1}) \in \C$ implies $(c_{n-1}, c_0, c_1, \cdots, c_{n-2}) 
\in \C$.  
By identifying any vector $(c_0,c_1, \cdots, c_{n-1}) \in \gf(q)^n$ 
with  
$$ 
c_0+c_1x+c_2x^2+ \cdots + c_{n-1}x^{n-1} \in \gf(q)[x]/(x^n-1), 
$$
any code $\C$ of length $n$ over $\gf(q)$ corresponds to a subset of the quotient ring 
$\gf(q)[x]/(x^n-1)$. 
A linear code $\C$ is cyclic if and only if the corresponding subset in $\gf(q)[x]/(x^n-1)$ 
is an ideal of the ring $\gf(q)[x]/(x^n-1)$. 

Note that every ideal of $\gf(q)[x]/(x^n-1)$ is principal. Let $\C=\langle g(x) \rangle$ be a 
cyclic code, where $g(x)$ is monic and has the smallest degree among all the 
generators of $\C$. Then $g(x)$ is unique and called the {\em generator polynomial,} 
and $h(x)=(x^n-1)/g(x)$ is referred to as the {\em parity-check} polynomial of $\C$. 

Let $n$ be a positive integer and 
let $\Z_n$ denote  the set $\{0,1,2, \cdots, n-1\}$.  Let $s$ be an integer with $0 \leq s <n$. The \emph{$q$-cyclotomic coset of $s$ modulo $n$\index{$q$-cyclotomic coset modulo $n$}} is defined by 
$$ 
C_s=\{s, sq, sq^2, \cdots, sq^{\ell_s-1}\} \bmod n \subseteq \Z_n,  
$$
where $\ell_s$ is the smallest positive integer such that $s \equiv s q^{\ell_s} \pmod{n}$, and is the size of the 
$q$-cyclotomic coset. The smallest integer in $C_s$ is called the \emph{coset leader\index{coset leader}} of $C_s$. 
Let $\Gamma_{(n,q)}$ be the set of all the coset leaders. We have then $C_s \cap C_t = \emptyset$ for any two 
distinct elements $s$ and $t$ in  $\Gamma_{(n,q)}$, and  
\begin{eqnarray}\label{eqn-cosetPP}
\bigcup_{s \in  \Gamma_{(n,q)} } C_s = \Z_n. 
\end{eqnarray}
Hence, the distinct $q$-cyclotomic cosets modulo $n$ partition $\Z_n$. 

Let $m=\ord_{n}(q)$ be the order of $q$ modulo $n$, and let $\alpha$ be a generator of $\gf(q^m)^*$. Put $\beta=\alpha^{(q^m-1)/n}$. 
Then $\beta$ is a primitive $n$-th root of unity in $\gf(q^m)$. The minimal polynomial $\m_{\beta^s}(x)$ 
of $\beta^s$ over $\gf(q)$ is the monic polynomial of the smallest degree over $\gf(q)$ with $\beta^s$ 
as a root.  It is straightforward to see that this polynomial is given by 
\begin{eqnarray}
\m_{\beta^s}(x)=\prod_{i \in C_s} (x-\beta^i) \in \gf(q)[x], 
\end{eqnarray} 
which is irreducible over $\gf(q)$. It then follows from (\ref{eqn-cosetPP}) that 
\begin{eqnarray}\label{eqn-canonicalfact}
x^n-1=\prod_{s \in  \Gamma_{(n,q)}} \m_{\beta^s}(x)
\end{eqnarray}
which is the factorization of $x^n-1$ into irreducible factors over $\gf(q)$. This canonical factorization of $x^n-1$ 
over $\gf(q)$ is crucial for the study of cyclic codes.

Let $\delta$ be an integer with $2 \leq \delta \leq n$ and let $h$ be an integer.  
A \emph{BCH code\index{BCH codes}} over $\gf(q)$ 
with length $n$ and \emph{designed distance} $\delta$, denoted by $\C_{(q,n,\delta,h)}$, is a cyclic code with 
generator polynomial 
\begin{eqnarray}\label{eqn-BCHdefiningSet}
g_{(q,n,\delta,h)}=\lcm(\m_{\beta^h}(x), \m_{\beta^{h+1}}(x), \cdots, \m_{\beta^{h+\delta-2}}(x)) 
\end{eqnarray}
where the least common multiple is computed over $\gf(q)$.

It may happen that $\C_{(q,n,\delta_1,h)}$ and $\C_{(q,n,\delta_2,h)}$ are identical for two distinct 
$\delta_1$ and $\delta_2$. The 
maximum designed distance of a BCH code is also called the \emph{Bose distance\index{Bose distance}}. 

When $h=1$, the code $\C_{(q,n,\delta,h)}$ with the generator polynomial in (\ref{eqn-BCHdefiningSet}) is called a \emph{narrow-sense\index{narrow sense}} BCH code. If $n=q^m-1$, then $\C_{(q,n,\delta,h)}$ is referred to as a \emph{primitive\index{primitive BCH}} BCH code. 

BCH codes are a subclass of cyclic codes with interesting properties. In many cases BCH codes are the best linear codes. 
For example, among all binary cyclic codes of odd length $n$ with $n \leq 125$ the best cyclic code is always a BCH code 
except for two special cases \cite{Dingbook15}. Reed-Solomon codes are also BCH codes and are widely used in communication 
devices and consumer electronics. In the past ten years, a lot of progress on the study of BCH codes has been made 
(see, for example, \cite{LWL19,LiSIAM,LLFLR,SYW,YLLY}).  

It is well known that the extended code $\overline{\C_{(q,q^m-1,\delta,1)}}$ of the narrow-sense primitive BCH code 
$\C_{(q,q^m-1,\delta,1)}$ holds $2$-designs, as the permutation automorphism group of the extended code contains 
the general affine group as a subgroup (see, for example, 
\cite{DingZhouConf17} and \cite[Chapter 8]{Dingbook18}). 
However, It is very rare that an infinite family of 
cyclic codes hold an infinite family of $3$-designs. 
In this paper, we will present an infinite family of BCH codes 
holding an infinite family of $4$-designs, which makes a breakthrough in 71 years.     

\subsection{AMDS codes and NMDS codes}

An $[n, k, n-k+1]$ linear code is called an MDS code. 
An $[n, k, n-k]$ linear code is said to be almost maximum distance separable (almost MDS or AMDS for short). 
A code is said to be near  maximum distance separable (near MDS or NMDS for short) if the code and its dual code 
both are almost maximum distance separable. MDS codes do hold $t$-designs with very large $t$. 
Unfortunately, all $t$-designs held in MDS codes are complete and thus trivial.  The first near MDS code was the $[11, 6, 5]$ ternary Golay code 
discovered in 1949 by Golay \cite{Golay49}. This ternary code holds $4$-designs, and its extended code holds a Steiner system 
$S(5, 6, 12)$ with the largest strength known. 
Ding and Tang very recently presented an infinite family of near MDS codes over $\gf(3^m)$ holding an infinite family of $3$-designs and an 
infinite family of near MDS  codes over $\gf(2^{2m})$ holding an infinite family of $2$-designs \cite{DingTang19}. 

NMDS codes have nice properties \cite{DodLan95,DodLan00,FaldumWillems97,TD13}. In particular, up to a multiple, there is a natural 
correspondence between the minimum weight codewords of an NMDS code $\C$ and its dual 
$\C^\perp$, which follows from the next result \cite{FaldumWillems97}. 

\begin{theorem}\label{thm-121FW}
Let $\C$ be an NMDS code. Then for every minimum weight codeword $\bc$ in $\C$, there exists, 
up to a multiple, a unique minimum weight codeword $\bc^\perp$ in $\C^\perp$ such that 
$\support(\bc) \cap \support(\bc^\perp)=\emptyset$. In particular, $\C$ and $\C^\perp$ 
have the same number of minimum weight codewords. 
\end{theorem} 

By Theorem \ref{thm-121FW}, if the minimum weight codewords of an NMDS code support a $t$-design, 
so do the minimum weight codewords of its dual, and the two $t$-designs are complementary of each other.

\section{Combinatorial  $t$-designs from elementary symmetric polynomials}

The objective of this section is to construct  $3$-designs and $4$-designs from elementary symmetric polynomials.
These results  would play a crucial role in proving that the codes constructed in the next section support $3$-designs or $4$-designs.

We define $[k]:= \{1,2, \cdots, k\}$. 
The \emph{elementary symmetric polynomial} (\emph{ESP}) of degree $\ell$ in $k$ variables $u_1, u_2, \cdots, u_k$, written  $\sigma_{k,\ell}$, is defined by
\begin{align}\label{eq:esp}
\sigma_{k,\ell}(u_1, \cdots, u_{k})= \sum_{I\subseteq [k], \# I=\ell}  \prod_{j\in I} u_j.
\end{align}
In commutative algebra, the elementary symmetric polynomials are a type of basic building block for symmetric polynomials, 
in the sense that any symmetric polynomial can be expressed as a polynomial in elementary symmetric polynomials.

Let $q=2^m$. Let $U_{q+1}$ be the subgroup of $\gf(q^2)^*$ of order $q+1$, that is, $U_{q+1}=\{u \in \gf(q^2)^*: u^{q+1}=1\}$. 
For any integer $k$ with $1 \leq k \leq q+1$, let $\binom{U_{q+1}}{k}$ denote the set of all $k$-subsets of $U_{q+1}$. 
 Define 
\begin{align}\label{eq:sp-B}
\cB_{\sigma_{k,\ell},q+1}=\left \{ \{u_1, \cdots, u_k\} \in \binom{U_{q+1}}{k} : \sigma_{k, \ell}(u_1, \cdots, u_k)=0 \right \}.
\end{align}
The incidence structure $\bD_{\sigma_{k,\ell},q+1}=(U_{q+1}, \cB_{\sigma_{k,\ell},q+1})$ may be a $t$-$(q+1,k,\lambda)$ design for some $\lambda$, where $U_{q+1}$ is the point set, 
and the incidence relation is the set membership. In this case, we say that the ESP $\sigma_{k,\ell}$ supports a $t$-$(q+1,k,\lambda)$ design.
The ESP $\sigma_{k,\ell}$ always supports a $1$-design,  but may not support $2$-designs.
Define the block sets $\cB_{\sigma_{6,3},q+1}^{0}$ and $\cB_{\sigma_{6,3},q+1}^{1}$ by
\begin{eqnarray}\label{eq:B0}
\cB_{\sigma_{6,3},q+1}^{0}=\left\{ 
\begin{array}{lr}
 \{u_1,u_2,u_3,u_4,u_5,u_6\} \in \cB_{\sigma_{6,3},q+1}:  &\{ u_{i_1},u_{i_2},u_{i_3},u_{i_4},u_{i_5}\} \in \cB_{\sigma_{5,2},q+1},\\
                                                        &1\le i_1<i_2<i_3<i_4<i_5\le 6
\end{array}
\right \},
\end{eqnarray} 
and 
\begin{align}\label{eq:B1}
\cB_{\sigma_{6,3},q+1}^{1}=\cB_{\sigma_{6,3},q+1} \setminus  \cB_{\sigma_{6,3},q+1}^{0}.
\end{align}

The following three theorems and corollary are the main results of this section. They show an interesting application of ESPs in the theory of combinatorial designs.

\begin{theorem}\label{thm:esp4-design}
Let $q=2^m$ with $m \ge  5$ odd. Then the incidence structure $(U_{q+1}, \cB_{\sigma_{6,3}, q+1})$ is a $4$-$\left (q+1, 6, \frac{q-8}{2} \right )$ design, where 
the block set $\cB_{\sigma_{6,3},q+1}$
is given by (\ref{eq:sp-B}). 
\end{theorem} 

\begin{theorem}\label{thm:espsteiner}
Let $q=2^m$ with $m \ge  4$ even. Then the incidence structure $(U_{q+1}, \cB_{\sigma_{5,2}, q+1})$ is a 
Steiner system $S(3,5,q+1)$, where  the block set $\cB_{\sigma_{5,2},q+1}$
is given by (\ref{eq:sp-B}). 
\end{theorem} 

\begin{theorem}\label{thm:esp3-design}
Let $q=2^m$ with $m \ge  4$ even. Then the incidence structure $(U_{q+1}, \cB^0_{\sigma_{6,3}, q+1})$ is a $3$-$\left (q+1, 6, 2(q-4) \right )$ design,
and the incidence structure $(U_{q+1}, \cB_{\sigma_{6,3}, q+1})$ is a $3$-$\left (q+1, 6, \frac{(q-4)^2}{6} \right )$ design. 
\end{theorem}

The following corollary follows immediately  from the previous theorem.

\begin{corollary}\label{cor:B1-WT6}
Let $q=2^m$ with $m \ge  4$ even. Then the incidence structure $(U_{q+1}, \cB^1_{\sigma_{6,3}, q+1})$ is a $3$-$\left (q+1, 6, \frac{(q-4)(q-16)}{6}\right )$ design.
\end{corollary} 

From Theorems \ref{thm:esp4-design}, \ref{thm:espsteiner} and \ref{thm:esp3-design}, one gets
\begin{eqnarray*}
\# \cB_{\sigma_{5,2}, q+1}=\left\{ 
\begin{array}{ll}
 \frac{1}{10} \binom{q+1}{3}, &  \text{ if } q=2^{2m}, \\
0, & \text{ if } q=2^{2m+1},
\end{array}
\right. 
\end{eqnarray*} 
and
\begin{eqnarray*}
\# \cB_{\sigma_{6,3}, q+1}=\left\{ 
\begin{array}{ll}
 \frac{(q-4)^2}{120} \binom{q+1}{3}, &  \text{ if } q=2^{2m}, \\
\frac{q-8}{30}\binom{q+1}{4}, & \text{ if } q=2^{2m+1}.
\end{array}
\right. 
\end{eqnarray*} 
In general, it's difficult to determine $\# \cB_{\sigma_{k,\ell}, q+1}$. It would be interesting to settle the following problem.
\begin{open}
Let $q=2^m$, and $k, \ell$ be two positive integers with  $\ell \le \frac{k}{2}$. 
Determine the cardinality of the block set $\cB_{\sigma_{k,\ell}, q+1}$ given by (\ref{eq:sp-B}) for $(k, \ell) \neq  (6,3) \text{ and } (5,2)$.

\end{open}

To prove Theorems \ref{thm:esp4-design}, \ref{thm:espsteiner}, and \ref{thm:esp3-design}, we need the following lemmas. 
The first one is on quadratic equations over finite fields of characteristic two \cite{LN97}, and is documented next. 

\begin{lemma}\label{lem:quadeq}
Let $f(T)=T^2+aT+b \in \gf(2^m)$ be a  polynomial of degree $2$. Then
\begin{enumerate}
\item $f$ has exactly one root in $\gf(2^m)$ if and only if $a=0$; 
\item $f$ has exactly two roots in $\gf(2^m)$ if and only if $a\neq 0$ and $\tr_{q/2}\left (\frac{b}{a^2} \right )=0$; and 
\item $f$ has exactly two roots in $\gf(2^{2m}) \setminus \gf(2^m)$ if and only if $a\neq 0$ and $\tr_{q/2}\left (\frac{b}{a^2}\right )=1$.
\end{enumerate}
\end{lemma}

\begin{lemma}\label{lem:tr=1}
Let $\{u_1,u_2\} \in \binom{U_{q+1}}{2}$. Then $\frac{u_1u_2}{u_1^2+u_2^2} \in  \gf(q)$ and $\tr_{q/2}\left ( \frac{u_1u_2}{u_1^2+u_2^2} \right )=1$.
\end{lemma}

\begin{proof}
Let $a=\frac{u_1u_2}{u_1^2+u_2^2}$. Then $a^q=\frac{u_1^{-1}u_2^{-1}}{u_1^{-2}+u_2^{-2}}=a$. Thus $a\in  \gf(q)$.

Note that $\frac{1}{a}=u +\frac{1}{u}$, where $u=\frac{u_1}{u_2} \in  U_{q+1}$. One has
\begin{align}\label{eqn-dec251}
(au)^2+(au)+a^2=0,
\end{align}
where $au \in \gf(q^2) \setminus \gf(q)$. Hence, the equation $T^2+T+a^2=0$ has two roots in $\gf(2^{2m}) \setminus \gf(2^m)$. 
It then follows from Lemma \ref{lem:quadeq} that $\tr(a)=\tr(a^2)=1$.  
This completes the proof.
\end{proof}

\begin{lemma}\label{lem:sigma41}
Let $q=2^m$ and $\{u_1,u_2,u_3,u_4\} \in \binom{U_{q+1}}{4}$. Then we have the following. 
\begin{enumerate}
\item $u_1+u_2+u_3+u_4 \neq 0$.
\item If  $m$ is even, then $u_1+u_2+u_3\neq 0$.
\end{enumerate}
\end{lemma}

\begin{proof}
Assume that $u_1+u_2+u_3+u_4=0$.  We have then
\begin{align*}
\frac{1}{u_1}+\frac{1}{u_2}+\frac{1}{u_3}+\frac{1}{u_4}=(u_1+u_2+u_3+u_4)^q=0.
\end{align*}
It follows from $u_4=u_1+u_2+u_3$ that 
\begin{align*}
\frac{1}{u_1}+\frac{1}{u_2}+\frac{1}{u_3}+\frac{1}{u_1+u_2+u_3}=0. 
\end{align*}
Multiplying  both sides of the previous equation by $u_1u_2u_3(u_1+u_2+u_3)$ yields
\begin{align*}
(u_1+u_2+u_3)(u_1u_2+u_2u_3+u_3u_1)+u_1u_2u_3=0,
\end{align*}
which is the same as
\begin{align*}
(u_1+u_2)(u_2+u_3)(u_3+u_1)=0,
\end{align*}
which  is contrary to our assumption that $u_1,u_2, u_3$ are pairwise distinct.
Thus, $u_1+u_2+u_3+u_4\neq 0$.

Let $q=2^m$ with $m$ even. Assume that $u_1+u_2+u_3=0$.
Then $\frac{1}{u_1+u_2}=\frac{1}{u_3}=\frac{1}{u_1}+\frac{1}{u_2}=\frac{u_1+u_2}{u_1u_2}$.
We then  have $u_1^2+u_1u_2+u_2^2=0$. Thus, $u_1^3=u_2^3$. Since $m$ is even , $\gcd(3, 2^m+1)=1$. 
It then follows from $u_1^3=u_2^3$ that $u_1=u_2$,  which  is contrary to our assumption that $u_1\neq u_2$.
This completes the proof.
\end{proof}

\begin{lemma}\label{lem:sigmasigma}
Let $\sigma_{3,1}, \sigma_{3,2}, \sigma_{3,3}$ be the ESPs given by (\ref{eq:esp}) with $\{u_1,u_2,u_3\} \in  \binom{U_{q+1}}{3}$. Then
\begin{enumerate}
\item  $\sigma_{3,1}\sigma_{3,2}+\sigma_{3,3}=(u_1+u_2)(u_2+u_3)(u_3+u_1)$.
\item $\sigma_{3,1}\sigma_{3,2}+\sigma_{3,3}\neq 0$.
\item $\sigma_{3,2}^2+\sigma_{3,1}\sigma_{3,3}= \sigma_{3,3}^2\left (\sigma_{3,1}^2+ \sigma_{3,2} \right )^q$.
\end{enumerate}
\end{lemma}
\begin{proof}
The proofs are straightforward and omitted.  
\end{proof}

\begin{lemma}\label{lem:sigma3not=0}
Let $q=2^m$ with $m$ even. Let $\sigma_{3,1}, \sigma_{3,2}, \sigma_{3,3}$ be the ESPs given by (\ref{eq:esp}) with 
$\{u_1,u_2,u_3\} \in  \binom{U_{q+1}}{3}$. Then
\begin{enumerate}
\item $\sigma_{3,1}^2+ \sigma_{3,2}\neq 0$; and 
\item $\sigma_{3,2}^2+\sigma_{3,1}\sigma_{3,3} \neq 0$.
\end{enumerate}
\end{lemma}
\begin{proof}
Assume that $\sigma_{3,1}^2+ \sigma_{3,2}= 0$, that is
\begin{align*}
u_1^2+u_2^2+u_3^2+u_1u_2+u_2u_3+u_3u_1=0.
\end{align*}
Multiplying both sides of previous equation by $u_1+u_2+u_3$ yields
\begin{align*}
u_1^3+u_2^3+u_3^3+u_1u_2u_3=0.
\end{align*}
It then follows that $\#\{u_1^3,u_2^3,u_3^3,u_1u_2u_3\}=3$ from Lemma \ref{lem:sigma41}, which is contrary to the  assumption that $m$ is even.
Combining Part 1 and  Lemma \ref{lem:sigmasigma} gives Part 2. 
This completes the proof.
\end{proof}

\begin{lemma}\label{lem:u4+u5}
 Let $u_j \in U_{q+1}$ such that $\sigma_{5,2}=0$, where $j \in \{1,2,3,4,5\}$.
Then
\begin{eqnarray*}
\left\{ 
\begin{array}{ll}
(\sigma_{3,1}^2+ \sigma_{3,2})(u_4+u_5)&=\sigma_{3,1}\sigma_{3,2}+\sigma_{3,3}, \\
(\sigma_{3,1}^2+ \sigma_{3,2})u_4u_5&=\sigma_{3,2}^2+\sigma_{3,1}\sigma_{3,3},  
\end{array}
\right. 
\end{eqnarray*} 
where $\sigma_{3,1}, \sigma_{3,2}, \sigma_{3,3}$ and $\sigma_{5,2}$ are the ESPs given by (\ref{eq:esp}). 
\end{lemma}

\begin{proof}
Let us observe first that
\begin{align}\label{eq:u4+u5}
u_4u_5+\sigma_{3,1} (u_4+u_5)+\sigma_{3,2}=0.
\end{align}
Raising to the $q$-th power both sides of Equation (\ref{eq:u4+u5})  yields
\begin{align*}
u_4^{-1}u_5^{-1}+\sigma_{3,1}^q(u_4^{-1}+u_{5}^{-1})+\sigma_{3,2}^{q}=0,
\end{align*}
which is the same as
\begin{align}\label{eq:u4+u5--2}
\sigma_{3,1} u_4 u_5 +\sigma_{3,2}(u_4+u_5)+\sigma_{3,3}=0.
\end{align}
The desired conclusion then follows from Equations (\ref{eq:u4+u5}) and (\ref{eq:u4+u5--2}). This completes the proof.
\end{proof}

\begin{lemma}\label{lem:one-one}
Let $q=2^m$ with $m$ even and $\{u_1,u_2,u_3,u_4,u_5,u_6\} \in \cB^{0}_{\sigma_{6,3}, q+1}$.
Let $A$ and $A'$ be two $5$-subsets of $\{u_1,u_2,u_3,u_4,u_5,u_6\}$ such that
$A,A' \in \cB_{\sigma_{5,2}, q+1}$. Then $A=A'$.
\end{lemma}

\begin{proof} 
Suppose that $A \neq A'$. Due to symmetry, let $A=\{u_1, u_2, u_3, u_4, u_5\} \in \cB_{\sigma_{5,2}, q+1}$ 
and $A'=\{u_1, u_2, u_3, u_4, u_6\} \in \cB_{\sigma_{5,2}, q+1}$. It then follows from Lemma  \ref{lem:u4+u5} that 
$$ 
(\sigma_{3,1}^2+ \sigma_{3,2})(u_4+u_5) =\sigma_{3,1}\sigma_{3,2}+\sigma_{3,3}=(\sigma_{3,1}^2+ \sigma_{3,2})(u_4+u_6), 
$$ 
which gives 
$$ 
(\sigma_{3,1}^2+ \sigma_{3,2})(u_5+u_6)=0.  
$$ 
It then follows from Lemma \ref{lem:sigma3not=0} that $u_5+u_6=0$, which is contrary to the assumption that 
$u_5 \neq u_6$. 
\end{proof}

The following follows immediately from Lemmas \ref{lem:sigmasigma},  \ref{lem:sigma3not=0}, and \ref{lem:u4+u5}.

\begin{lemma}\label{lem:all-not-0}
 Let $\{u_1, u_2, u_3\} \in \binom{U_{q+1}}{3}$ and $u_4,u_5 \in U_{q+1}$ such that $\sigma_{5,2}=0$.
Then none of 
$\sigma_{3,1}^2+ \sigma_{3,2}, \sigma_{3,1}\sigma_{3,2}+\sigma_{3,3}$ and $\sigma_{3,2}^2+\sigma_{3,1}\sigma_{3,3}$
equals zero, and $u_4\neq u_5$.
\end{lemma}

\begin{lemma}\label{lem:tr(b/a2)}
Let $q=2^m$ and $\{u_1, u_2, u_3\} \in \binom{U_{q+1}}{3}$ such that 
$(\sigma_{3,1}^2+ \sigma_{3,2})(\sigma_{3,1}\sigma_{3,2}+\sigma_{3,3})(\sigma_{3,2}^2+\sigma_{3,1}\sigma_{3,3})\neq  0$.
 Put $a=\frac{\sigma_{3,1}\sigma_{3,2}+\sigma_{3,3}}{\sigma_{3,1}^2+ \sigma_{3,2}}$ 
and $b=\frac{\sigma_{3,2}^2+\sigma_{3,1}\sigma_{3,3}}{\sigma_{3,1}^2+ \sigma_{3,2}}$. Then
$b\in U_{q+1}, \frac{b}{a^2} \in \gf(q)$ and $\tr_{q/2}\left ( \frac{b}{a^2} \right )\equiv 1+ m \pmod 2$.
\end{lemma}

\begin{proof}
First, $b\in  U_{q+1}$ follows from Part 3 of Lemma \ref{lem:sigmasigma}.
Next, observe that
\begin{align}\label{eq:12+23+31+0}
\frac{b}{a^2}= \frac{u_1u_2}{(u_1+u_2)^2}+\frac{u_2u_3}{(u_2+u_3)^2}+\frac{u_3u_1}{(u_3+u_1)^2}+1.
\end{align}
The desired conclusion then follows from Lemma  \ref{lem:tr=1} and Equation (\ref{eq:12+23+31+0}).
This completes the proof.
\end{proof}

\begin{lemma}\label{lem:quadeq-U}
Let the notation and assumption be the same as in Lemma \ref{lem:tr(b/a2)}.
Let $f(u)$ be   the  quadratic polynomial $u^2+au+b$. Then we have the following. 
\begin{enumerate}
\item If $m$ is odd, $f$ has no root in $U_{q+1} \setminus \left \{\sqrt{b} \right \}$.
\item If $m$ is even, $f$ has exactly two roots in $U_{q+1}$.
\end{enumerate}
\end{lemma} 

\begin{proof}
Let $m$ be odd. Assume that there exists an $u\in U_{q+1} \setminus \left \{\sqrt{b} \right \}$ such  that $f(u)=0$.
Then
\begin{align*}
\left (\frac{u}{\sqrt{b}} \right )^2+ \frac{a}{\sqrt{b}} \left (\frac{u}{\sqrt{b}} \right )+1=0. 
\end{align*}
From Lemma \ref{lem:quadeq} and $\frac{u}{\sqrt{b}} \in U_{q+1} \setminus \{1\} \subseteq \gf(q^2) \setminus \gf(q)$,
$\tr_{q/2}\left ( \frac{b}{a^2} \right )=1$, which is contrary to the result of Lemma \ref{lem:tr(b/a2)}.

Let  $m$ be even. From Lemmas \ref{lem:quadeq} and  \ref{lem:tr(b/a2)}, there  exists $u'\in \gf(q^2) \setminus  \gf(q)$
such that $u', u'^q$ are exactly  the two solutions   of the 
quadratic equation $T^2+\frac{a}{\sqrt{b}} T+1=0$. 
It's easily checked that $u_4=\sqrt{b}u'$ and $u_5=\sqrt{b}u'^q$ are the two roots  of $f$.
The desired conclusion then  follows from the fact $u'^{q+1}=1$.
This  completes the proof.
\end{proof}

Combining Lemmas \ref{lem:all-not-0}, \ref{lem:u4+u5},  and \ref{lem:quadeq-U} gives the following.
\begin{lemma}\label{lem:B52odd}
Let $q=2^m$ with $m$ odd and $\{u_1,u_2,u_3,u_4,u_5\} \in \binom{U_{q+1}}{5}$. Then
$ 
\sigma_{5,2} \neq 0.
$ 

\end{lemma}

\begin{lemma}\label{lem:quadeq-u4u5}
Let $q=2^m$ with $m$ even and $\{u_1, u_2, u_3\} \in \binom{U_{q+1}}{3}$. Let $u_4, u_5$ be
the two solutions  of   the  quadratic equation $u^2+au+b=0$, where $a=\frac{\sigma_{3,1}\sigma_{3,2}+\sigma_{3,3}}{\sigma_{3,1}^2+ \sigma_{3,2}}$ 
and $b=\frac{\sigma_{3,2}^2+\sigma_{3,1}\sigma_{3,3}}{\sigma_{3,1}^2+ \sigma_{3,2}}$. Then 
$$\{u_1,u_2,u_3,u_4,u_5\} \in \cB_{\sigma_{5,2},q+1}.$$
\end{lemma}

\begin{proof}
First, employing Lemmas \ref{lem:sigmasigma}, \ref{lem:sigma3not=0},
and \ref{lem:quadeq-U}, we have that $u_4, u_5 \in U_{q+1}$ and $u_4\neq u_5$.
Using $\sigma_{5,2}=u_4 u_5+(u_4+u_5) \sigma_{3,1}+\sigma_{3,2}$ and Vieta's formulas yields
\begin{align*}
\sigma_{5,2}= \frac{\sigma_{3,2}^2+\sigma_{3,1}\sigma_{3,3}}{\sigma_{3,1}^2+ \sigma_{3,2}}+ \frac{\sigma_{3,1}\sigma_{3,2}+
\sigma_{3,3}}{\sigma_{3,1}^2+ \sigma_{3,2}} \sigma_{3,1}+\sigma_{3,2} 
= 0.
\end{align*}

Suppose that $u_4=u_i$ and $u_5=u_j$ for some $i,j \in \{1,2,3\}$.  By symmetry, let $(i,j)=(3,2)$. Then
\begin{align*}
\sigma_{5,2}= u_3u_4+u_2u_5 
= u_2^2+u_3^2 
= 0,
\end{align*}
which is contrary to $u_2 \neq u_3$.
Thus, $\#\left  ( \{u_1, u_2, u_3\} \cap \{u_4, u_5\}  \right ) \neq 2$.

Assume that $\#\left  ( \{u_1, u_2, u_3\} \cap \{u_4, u_5\}  \right ) = 1$.
By the symmetry of $u_1,u_2,u_3$, let $u_5=u_3$ and $u_4\not \in \{u_1, u_2, u_3\}$.
Then $\sigma_{5,2}(u_1,u_2,u_4, u_5, u_3)=0$. Note that
 $\{u_1,u_2, u_4\} \in  \binom{U_{q+1}}{3}$ and $u_5=u_3$, which is contrary to Lemma \ref{lem:all-not-0}.
 Thus, $\#\left  ( \{u_1, u_2, u_3\} \cap \{u_4, u_5\}  \right ) \neq   1$.
 Hence, $\{u_1,u_2,u_3,u_4,u_5\} \in \binom{U_{q+1}}{5}$.
 This completes the proof.
\end{proof}

\begin{lemma}\label{lem:sigma4notzero}
Let  $\{u_1,u_2,u_3,u_4 \} \in \binom{U_{q+1}}{4}$. Then $\sigma_{4,3} \sigma_{4,1} \neq 0$ and $(\sigma_{4,3}+u_i \sigma_{4,2})(\sigma_{4,2}+u_i \sigma_{4,1})\neq 0$,
where $i\in \{1,2,3,4\}.$
\end{lemma}

\begin{proof}
Note that
\begin{align*}
\sigma_{4,3} \sigma_{4,1}=\sigma_{4,4}\sigma_{4,1}^{q+1}.
\end{align*}
By Part 1 of  Lemma \ref{lem:sigma41}, $\sigma_{4,3} \sigma_{4,1} \neq 0$.

Note that $(\sigma_{4,3}+u_i \sigma_{4,2})(\sigma_{4,2}+u_i \sigma_{4,1})= u_i \sigma_{4,4}(\sigma_{4,2}+u_i \sigma_{4,1})^{q+1}$.
We only need to prove that $\sigma_{4,2}+u_i \sigma_{4,1}\neq 0$. On  the contrary, suppose that $\sigma_{4,2}+u_i \sigma_{4,1}= 0$.
Using the symmetry of $u_1, u_2,u_3,u_4$, choose $u_i=u_4$.
Then $\sigma_{3,2}+u_4^2=u_1u_2+u_2u_3+u_3u_1+u_4^2=0$, which is contrary to Part 1 of Lemma \ref{lem:sigma41} 
if $u_4^2 \not\in \{u_1u_2, u_2u_3, u_3u_1\}$.  If $u_4^2  \in \{u_1u_2, u_2u_3, u_3u_1\}$, due to symmetry assume that 
$u_4^2=u_1u_2$. It then follows from  $u_1u_2+u_2u_3+u_3u_1+u_4^2=0$ that $u_1=u_2$, which is contradictory to the 
assumption that $u_1 \neq u_2$.  
This completes the proof.
\end{proof}

From Lemma \ref{lem:sigma4notzero},  The following is easily checked.

\begin{lemma}\label{lem:sigma4-5U}
Let  $\{u_1,u_2,u_3,u_4 \} \in \binom{U_{q+1}}{4}$. 
Then $\sqrt{\frac{\sigma_{4,3}}{\sigma_{4,1}}}, \frac{\sigma_{4,3}+u_i \sigma_{4,2}}{\sigma_{4,2}+u_i \sigma_{4,1}} \in U_{q+1}$,
where $i\in \{1,2,3,4\}.$
\end{lemma}

\begin{lemma}\label{lem:sigma4-5Sigma631}
Let  $\{u_1,u_2,u_3,u_4 \} \in \binom{U_{q+1}}{4}$. 
Then $\sigma_{6,3}\left (u_1,u_2,u_3,u_4, \sqrt{\frac{\sigma_{4,3}}{\sigma_{4,1}}}, \sqrt{\frac{\sigma_{4,3}}{\sigma_{4,1}}} \right )=0$ and
$$ \sigma_{6,3}\left (u_1,u_2,u_3,u_4,\frac{\sigma_{4,3}+u_i \sigma_{4,2}}{\sigma_{4,2}+u_i \sigma_{4,1}}, u_i\right )=0,$$
where $i\in \{1,2,3,4\}.$
\end{lemma}

\begin{proof}
Set $u_5=u_6= \sqrt{\frac{\sigma_{4,3}}{\sigma_{4,1}}}$. Then
\begin{align*}
\sigma_{6,3}\left (u_1,u_2,u_3,u_4, u_5, u_6 \right )=& \sigma_{4,3} +(u_5+u_6) \sigma_{4,2} +u_5 u_6 \sigma_{4,1}\\
=& \sigma_{4,3} + u_5^2 \sigma_{4,1}\\
=& 0.
\end{align*}
Thus,  $\sigma_{6,3}\left (u_1,u_2,u_3,u_4, \sqrt{\frac{\sigma_{4,3}}{\sigma_{4,1}}}, \sqrt{\frac{\sigma_{4,3}}{\sigma_{4,1}}} \right )=0$.

Choose $\sigma_5=\frac{\sigma_{4,3}+u_i \sigma_{4,2}}{\sigma_{4,2}+u_i \sigma_{4,1}}$
and $\sigma_6=u_i$. Then
\begin{align*}
\sigma_{6,3}=& \sigma_{4,3} +(u_5+u_6) \sigma_{4,2} +u_5 u_6 \sigma_{4,1}\\
=& \sigma_{4,3} +\left ( \frac{\sigma_{4,3}+u_i \sigma_{4,2}}{\sigma_{4,2}+u_i \sigma_{4,1}} + u_i \right ) \sigma_{4,2} 
+\frac{\sigma_{4,3}+u_i \sigma_{4,2}}{\sigma_{4,2}+u_i \sigma_{4,1}}u_i\sigma_{4,1}\\
=& 0.
\end{align*}
This completes the proof.
\end{proof}

\begin{lemma}\label{lem:sigma4-5Sigma63}
Let  $\{u_1,u_2,u_3,u_4 \} \in \binom{U_{q+1}}{4}$ such that $\sigma_{5,2}(u_1,u_2,u_3,u_4,u_5)\neq 0$ for any $u_5 \in U_{q+1} \setminus \{u_1,u_2,u_3,u_4\}$. 
Let $S$ be the subset of $U_{q+1}$ given by
$$\left \{ \frac{\sigma_{4,3}+u_i \sigma_{4,2}}{\sigma_{4,2}+u_i \sigma_{4,1}}: i=1,2,3,4 \right \} \bigcup
\left  \{u_i: i=1,2,3,4 \right \} \bigcup \left \{\sqrt{\frac{\sigma_{4,3}}{\sigma_{4,1}}}  \right \}.$$
Then $\# S =9$.
\end{lemma}

\begin{proof}
First, we prove that $\sqrt{\frac{\sigma_{4,3}}{\sigma_{4,1}}}\neq u_4$.
On contrary, assume that $\sqrt{\frac{\sigma_{4,3}}{\sigma_{4,1}}}= u_4$. Then
$$\sigma_{4,1} u_4^2+\sigma_{4,3}=0,$$
which is the same as
$$u_4^3+ \sigma_{3,1} u_4^2 +\sigma_{3,2} u_4+ \sigma_{3,3} =0.$$
Then,
$$(u_4+u_1)(u_4+u_2)(u_4+u_3)=0,$$
which is contrary to the assumption $\{u_1,u_2,u_3,u_4 \} \in \binom{U_{q+1}}{4}$.
Thus $\sqrt{\frac{\sigma_{4,3}}{\sigma_{4,1}}}\neq u_4$. By the symmetry of $u_1, u_2, u_3, u_4$, 
\begin{align}\label{eq:2-3}
\sqrt{\frac{\sigma_{4,3}}{\sigma_{4,1}}}\neq u_i \mbox{ for all } i.
\end{align}

Assume that $\frac{\sigma_{4,3}+u_4 \sigma_{4,2}}{\sigma_{4,2}+u_4 \sigma_{4,1}}= u_4$.
Then $u_4 = \sqrt{\frac{\sigma_{4,3}}{\sigma_{4,1}}}$, which is contrary to Inequality (\ref{eq:2-3}).
Thus, $\frac{\sigma_{4,3}+u_4 \sigma_{4,2}}{\sigma_{4,2}+u_4 \sigma_{4,1}}\neq  u_4$. By the symmetry  of $u_1,u_2,u_3,u_4$,  
\begin{align}\label{eq:1-2ii}
\frac{\sigma_{4,3}+u_i \sigma_{4,2}}{\sigma_{4,2}+u_i \sigma_{4,1}}\neq  u_i  \mbox{ for all } i.
\end{align}

Assume that $\frac{\sigma_{4,3}+u_4 \sigma_{4,2}}{\sigma_{4,2}+u_4 \sigma_{4,1}}= u_3$.
Then $\sigma_{4,3}+u_4 \sigma_{4,2}+u_3(\sigma_{4,2}+u_4 \sigma_{4,1})=0$, which is the same as $(u_3+u_4)^2(u_1+u_2)= 0$.
This is contrary to our assumption $\{u_1,u_2,u_3,u_4\} \in \binom{U_{q+1}}{4}$.
Thus, $\frac{\sigma_{4,3}+u_4 \sigma_{4,2}}{\sigma_{4,2}+u_4 \sigma_{4,1}}\neq  u_3$. By the symmetry  of $u_1,u_2,u_3,u_4$,  
\begin{align}\label{eq:1-2ij}
\frac{\sigma_{4,3}+u_i \sigma_{4,2}}{\sigma_{4,2}+u_i \sigma_{4,1}}\neq  u_j \text{ for all } i \neq j.
\end{align}

Assume that $\frac{\sigma_{4,3}+u_i \sigma_{4,2}}{\sigma_{4,2}+u_i \sigma_{4,1}}= \sqrt{\frac{\sigma_{4,3}}{\sigma_{4,1}}}$ for some $i \in  \{1,2,3,4\}$.
Put $u_5= \sqrt{\frac{\sigma_{4,3}}{\sigma_{4,1}}}$.  From Inequality (\ref{eq:2-3}), $u_5\not \in \{u_1,u_2,u_3,u_4\}$.
By Lemma \ref{lem:sigma4-5Sigma631}, we have
\begin{eqnarray*}
\left\{ 
\begin{array}{ll}
 \sigma_{6,3}\left (u_1,u_2,u_3,u_4,u_5, u_i\right )&=0, \\
\sigma_{6,3}\left (u_1,u_2,u_3,u_4, u_5, \sqrt{\frac{\sigma_{4,3}}{\sigma_{4,1}}} \right )&=0. 
\end{array}
\right. 
\end{eqnarray*} 
By the assumption of this lemma, $\sigma_{5,2}(u_1,u_2,u_3,u_4,u_5)\neq 0$. Thus,
\begin{eqnarray*}
\left\{ 
\begin{array}{ll}
 u_i&=\frac{\sigma_{5,3}}{\sigma_{5,2}}, \\
\sqrt{\frac{\sigma_{4,3}}{\sigma_{4,1}}} &=\frac{\sigma_{5,3}}{\sigma_{5,2}}, 
\end{array}
\right. 
\end{eqnarray*} 
which is contrary to Inequality (\ref{eq:2-3}).
Hence, 
\begin{align}\label{eq:1-3}
\frac{\sigma_{4,3}+u_i \sigma_{4,2}}{\sigma_{4,2}+u_i \sigma_{4,1}}\neq  \sqrt{\frac{\sigma_{4,3}}{\sigma_{4,1}}}.
\end{align}

Assume that $\frac{\sigma_{4,3}+u_i \sigma_{4,2}}{\sigma_{4,2}+u_i \sigma_{4,1}}= \frac{\sigma_{4,3}+u_j \sigma_{4,2}}{\sigma_{4,2}+u_j \sigma_{4,1}}$ 
for some $i,j \in  \{1,2,3,4\}$.
Put $u_5= \frac{\sigma_{4,3}+u_i \sigma_{4,2}}{\sigma_{4,2}+u_i \sigma_{4,1}}$.  From Inequalities (\ref{eq:1-2ii}) and (\ref{eq:1-2ij}), $u_5\not \in \{u_1,u_2,u_3,u_4\}$.
By Lemma \ref{lem:sigma4-5Sigma631}, we have
\begin{eqnarray*}
\left\{ 
\begin{array}{ll}
 \sigma_{6,3}\left (u_1,u_2,u_3,u_4,u_5, u_i\right )&=0, \\
\sigma_{6,3}\left (u_1,u_2,u_3,u_4, u_5, u_j \right )&=0. 
\end{array}
\right. 
\end{eqnarray*} 
By the assumption of this lemma, $\sigma_{5,2}(u_1,u_2,u_3,u_4,u_5)\neq 0$. Thus,
\begin{eqnarray*}
\left\{ 
\begin{array}{ll}
 u_i&=\frac{\sigma_{5,3}}{\sigma_{5,2}}, \\
u_j &=\frac{\sigma_{5,3}}{\sigma_{5,2}}. 
\end{array}
\right. 
\end{eqnarray*} 
Then, $i=j$. Hence,
\begin{align}\label{eq:1-1}
\frac{\sigma_{4,3}+u_i \sigma_{4,2}}{\sigma_{4,2}+u_i \sigma_{4,1}}\neq  \frac{\sigma_{4,3}+u_j \sigma_{4,2}}{\sigma_{4,2}+u_j \sigma_{4,1}}, \text{ for }  i\neq j.
\end{align}

The desired conclusion then follows from Inequalities  (\ref{eq:2-3}), (\ref{eq:1-2ii}), (\ref{eq:1-2ij}), (\ref{eq:1-3}) and (\ref{eq:1-1}).
This completes the proof.

\end{proof}

\begin{lemma}\label{lem:4->5rigid}
Let $q=2^m$ with $m$ even. Let  $\{u_1',u_2',u_3', u_4', u_5'\} \in \cB_{\sigma_{5,2},q+1}$ and $u_5, u_6\in U_{q+1}$ such that
$\sigma_{6,3}(u_1',u_2',u_3', u_4', u_5,u_6)=0$. Then $u_5' \in \{u_5, u_6\}$.
\end{lemma}
\begin{proof}
Assume that $ u_5' \not \in \{u_5, u_6\}$.  By  Lemmas \ref{lem:sigma3not=0} and \ref{lem:u4+u5}, $\sigma_{5,2}(u_1',u_2',u_3',u_4',u_5) \neq 0$.
One has
\begin{eqnarray*}
\left\{ 
\begin{array}{ll}
 \sigma_{6,3}\left (u_1',u_2',u_3',u_4',u_5, u_5'\right )&=0, \\
\sigma_{6,3}\left (u_1',u_2',u_3',u_4', u_5, u_6 \right )&=0,
\end{array}
\right. 
\end{eqnarray*} 
which is the same as
\begin{eqnarray*}
\left\{ 
\begin{array}{ll}
 u_5'&= \frac{\sigma_{5,3}(u_1',u_2',u_3',u_4',u_5) }{\sigma_{5,2}(u_1',u_2',u_3',u_4',u_5) }, \\
u_6 &=\frac{\sigma_{5,3}(u_1',u_2',u_3',u_4',u_5) }{\sigma_{5,2}(u_1',u_2',u_3',u_4',u_5) }.
\end{array}
\right. 
\end{eqnarray*} 
This is contrary to our assumption that $ u_5' \not \in \{u_5, u_6\}$. This completes the proof.
\end{proof}

\begin{lemma}\label{lem:how5}
Let  $\{u_1,u_2,u_3,u_4 \} \in \binom{U_{q+1}}{4}$ such that $\sigma_{5,2}(u_1,u_2,u_3,u_4,u_5)\neq 0$ for any $u_5 \in U_{q+1} \setminus \{u_1,u_2,u_3,u_4\}$. 
Then
\begin{align*}
\frac{\sigma_{5,3}\left (u_1,u_2,u_3,u_4, \sqrt{\frac{\sigma_{4,3}}{\sigma_{4,1}}} \right ) }
{\sigma_{5,2} \left (u_1,u_2,u_3,u_4,\sqrt{\frac{\sigma_{4,3}}{\sigma_{4,1}}} \right ) }= \sqrt{\frac{\sigma_{4,3}}{\sigma_{4,1}}},
\end{align*}
and
\begin{align*}
\frac{\sigma_{5,3} \left (u_1,u_2,u_3,u_4, \frac{\sigma_{4,3}+u_i \sigma_{4,2}}{\sigma_{4,2}+u_i \sigma_{4,1}} \right ) }
{\sigma_{5,2}\left (u_1,u_2,u_3,u_4,\frac{\sigma_{4,3}+u_i \sigma_{4,2}}{\sigma_{4,2}+u_i \sigma_{4,1}} \right ) }= u_i,
\end{align*}
where $i\in \{1,2,3,4\}$.
\end{lemma}

\begin{proof}
The desired conclusion  then follows from Lemma \ref{lem:sigma4-5Sigma631}.
\end{proof}

We will need the following lemma whose proof is straightforward.
\begin{lemma}\label{lem:whatprob5}
Let  $\{u_1,u_2,u_3,u_4 \} \in \binom{U_{q+1}}{4}$ and $u_5 \in U_{q+1}$ such that  $\sigma_{5,2}\left (u_1,u_2,u_3,u_4,u_5 \right )\neq 0$.
Let $u_6= \frac{\sigma_{5,3}\left (u_1,u_2,u_3,u_4,u_5 \right )}{\sigma_{5,2}\left (u_1,u_2,u_3,u_4,u_5 \right ) }$.
Then we have the following. 
\begin{enumerate}
\item If $u_6=u_5$, then $u_5=  \sqrt{\frac{\sigma_{4,3}}{\sigma_{4,1}}} $.
\item If $u_6=u_i$, then $u_5= \frac{\sigma_{4,3}+u_i \sigma_{4,2}}{\sigma_{4,2}+u_i \sigma_{4,1}}$, where  $i \in \{1,2,3,4\}$.
\end{enumerate}
\end{lemma}

\begin{lemma}\label{lem:S11}
Let $q=2^m$ with $m$ even and  $\{u_1,u_2,u_3,u_4 \} \in \binom{U_{q+1}}{4}$ such that $\sigma_{5,2}(u_1,u_2,u_3,u_4,u_5)\neq 0$
 for any $u_5 \in U_{q+1} \setminus \{u_1,u_2,u_3,u_4\}$. 
Let $S$ be the subset of $U_{q+1}$ given by
$$\left \{ \frac{\sigma_{4,3}+u_i \sigma_{4,2}}{\sigma_{4,2}+u_i \sigma_{4,1}}: i=1,2,3,4 \right \} \bigcup
\left  \{u_i: i=1,2,3,4 \right \} \bigcup \left \{\sqrt{\frac{\sigma_{4,3}}{\sigma_{4,1}}}  \right \}.$$
Let $\tilde{u}_4$ and $\tilde{u}_5$   be the two solutions of the quadratic equation $u^2+au+b=0$, where $a=\frac{\sigma_{3,1}\sigma_{3,2}+\sigma_{3,3}}{\sigma_{3,1}^2+ \sigma_{3,2}}$ 
and $b=\frac{\sigma_{3,2}^2+\sigma_{3,1}\sigma_{3,3}}{\sigma_{3,1}^2+ \sigma_{3,2}}$.
Then $\tilde{u}_4 \not \in S$ and $\tilde{u}_5\not \in S$.
\end{lemma}

\begin{proof}
By the definition of $\tilde{u}_4, \tilde{u}_5$ and Lemma \ref{lem:u4+u5}, $u_4\not \in \{\tilde{u}_4, \tilde{u}_5\}$.
Suppose that $\tilde{u}_4= \sqrt{\frac{\sigma_{4,3}}{\sigma_{4,1}}} $.
From Lemma \ref{lem:sigma4-5Sigma631} or \ref{lem:how5}, one gets
\begin{align*}
 \sigma_{6,3}\left (u_1,u_2,u_3,u_4,\tilde{u}_4, \sqrt{\frac{\sigma_{4,3}}{\sigma_{4,1}}}\right )=0.
\end{align*}
From Lemma  \ref{lem:4->5rigid} and $\tilde{u}_5 \neq u_4$, $\tilde{u}_5=\sqrt{\frac{\sigma_{4,3}}{\sigma_{4,1}}}= \tilde{u}_4$,
which is contrary to $a\neq 0$. Thus, $\tilde{u}_4\neq  \sqrt{\frac{\sigma_{4,3}}{\sigma_{4,1}}} $. By the symmetry of $\tilde{u}_4$ and $ \tilde{u}_5$,
$\tilde{u}_5\neq  \sqrt{\frac{\sigma_{4,3}}{\sigma_{4,1}}} $.

Suppose that $\tilde{u}_4= \frac{\sigma_{4,3}+u_i \sigma_{4,2}}{\sigma_{4,2}+u_i \sigma_{4,1}} $.
From Lemma \ref{lem:sigma4-5Sigma631} or \ref{lem:how5}, one gets
\begin{align*}
 \sigma_{6,3}\left (u_1,u_2,u_3,u_4,u_i, \tilde{u}_4 \right )=0.
\end{align*}
From Lemma  \ref{lem:4->5rigid} and $\tilde{u}_5 \neq u_4$, $\tilde{u}_5= u_i$,
which is contrary to the definition of $\tilde{u}_5$. Thus, $\tilde{u}_4\neq \frac{\sigma_{4,3}+u_i \sigma_{4,2}}{\sigma_{4,2}+u_i \sigma_{4,1}}$. By
the symmetry of $\tilde{u}_4$ and $ \tilde{u}_5$,
$\tilde{u}_5\neq \frac{\sigma_{4,3}+u_i \sigma_{4,2}}{\sigma_{4,2}+u_i \sigma_{4,1}} $.
This completes the proof.
\end{proof}

\noindent \emph{Proof of Theorem \ref{thm:esp4-design}.} 
Let $\{u_1,u_2,u_3,u_4\}$ be a fixed $4$-subset of $U_{q+1}$. Set 
$$S=\left \{ \frac{\sigma_{4,3}+u_i \sigma_{4,2}}{\sigma_{4,2}+u_i \sigma_{4,1}}: i=1,2,3,4 \right \} \bigcup
\left  \{u_i: i=1,2,3,4 \right \} \bigcup \left \{\sqrt{\frac{\sigma_{4,3}}{\sigma_{4,1}}}  \right \}.$$
For any $u_5 \not \in \{u_i: i=1,2,3,4 \}$, $\sigma_{5,2}(u_1,u_2, u_3, u_4, u_5)\neq 0$ from Lemma \ref{lem:B52odd}.
Define
$$\mathcal T=\left \{ \left \{ u_5, \frac{\sigma_{5,3}(u_1,u_2, u_3, u_4, u_5)}{\sigma_{5,2}(u_1,u_2, u_3, u_4, u_5)} \right \}  : u_5 \in U_{q+1} \setminus S \right \}.$$
From Lemmas \ref{lem:how5} and \ref{lem:whatprob5}, $\frac{\sigma_{5,3}(u_1,u_2, u_3, u_4, u_5)}{\sigma_{5,2}(u_1,u_2, u_3, u_4, u_5)} \not \in S$ if $u_5 \not \in S$.
By Lemma \ref{lem:sigma4-5Sigma63}, $\# \mathcal T =\frac{(q+1-9)}{2}$.
From Lemma \ref{lem:whatprob5} and $\frac{\sigma_{5,3}(u_1,u_2, u_3, u_4, u_5)}{\sigma_{5,2}(u_1,u_2, u_3, u_4, u_5)} \in U_{q+1}$, 
we deduce that $\{u_1,u_2, u_3, u_4, u_5,u_6\} \in \cB_{\sigma_{6,3},q+1}$ for any $\{u_5, u_6\} \in \mathcal T$.
 
On the other hand, let $\{u_1,u_2, u_3, u_4, u_5,u_6\} \in \cB_{\sigma_{6,3},q+1}$. Employing   Lemma \ref{lem:how5}, $\{u_5, u_6\} \in \mathcal T$.
Thus, $\{u_1,u_2, u_3, u_4, u_5,u_6\} \in \cB_{\sigma_{6,3},q+1}$ if and only if $\{u_5, u_6\} \in \mathcal T$.
Hence, $(U_{q+1}, \cB_{\sigma_{6,3}, q+1})$ is a $4$-$\left (q+1, 6, \frac{q-8}{4} \right )$ design.  This completes the proof. \\
\rightline{$\square$}

\noindent \emph{Proof of Theorem \ref{thm:espsteiner}.} 
Let $\{u_1,u_2,u_3\}$ be a fixed $3$-subset of $U_{q+1}$. 
Employing Lemmas \ref{lem:u4+u5} and \ref{lem:quadeq-u4u5}, $\{u_1,u_2, u_3, u_4, u_5\} \in \cB_{\sigma_{6,3},q+1}$ if and only if $u_4$ and $u_5$ are the two 
solutions of the quadratic equation $u^2+au+b=0$ in $U_{q+1}$, where $a=\frac{\sigma_{3,1}\sigma_{3,2}+\sigma_{3,3}}{\sigma_{3,1}^2+ \sigma_{3,2}}$ 
and $b=\frac{\sigma_{3,2}^2+\sigma_{3,1}\sigma_{3,3}}{\sigma_{3,1}^2+ \sigma_{3,2}}$.
Hence, $(U_{q+1}, \cB_{\sigma_{5,2}, q+1})$ is a Steiner System $S(3,5,q+1)$.  This completes the proof.\\
\rightline{$\square$}

\noindent \emph{Proof of Theorem \ref{thm:esp3-design}.} 
For any   $3$-subset $\{u_1,u_2,u_3\}$ of $U_{q+1}$, let $Q(u_1,u_2,u_3)$ denote the $2$-subset $\left \{u\in U_{q+1}: u^2+au+b=0 \right \}$, where
 $a=\frac{\sigma_{3,1}\sigma_{3,2}+\sigma_{3,3}}{\sigma_{3,1}^2+ \sigma_{3,2}}$ 
and $b=\frac{\sigma_{3,2}^2+\sigma_{3,1}\sigma_{3,3}}{\sigma_{3,1}^2+ \sigma_{3,2}}$. 
Next, let $\{u_1,u_2,u_3\}$ be fixed.
Set
\begin{align*}
\mathcal T^0_{1}= \left \{ S^0 \cup \{u_6\}: u_6 \in U_{q+1} \setminus S^0 \right \},
\end{align*}
and
\begin{align*}
\mathcal T^0_{i,j}= \left \{ \{u_1,u_2,u_3,u_4\} \cup Q(u_i,u_j,u_4): u_4 \in U_{q+1} \setminus S^0 \right \},
\end{align*}
where $1\le i <  j \le 3$ and $S^0= \{u_1,u_2,u_3\} \cup Q(u_1,u_2,u_3)$.
Let $\mathcal T^0= \mathcal T^0_{1} \cup \mathcal T^0_{1,2} \cup \mathcal T^0_{1,3} \cup \mathcal T^0_{2,3}$.
It is easily checked that $\{u_1,u_2,u_3,u_4,u_5,u_6\} \in \cB^0_{\sigma_{6,3},q+1}$ if and only
if $\{u_1,u_2,u_3,u_4,u_5,u_6\} \in \mathcal T^0$.
Note that $\# \mathcal T^0_{1}= q-4 $ and $\# \mathcal T^0_{i,i}= \frac{q-4}{3}$, where $1\le i < j \le 3$.
From Lemma \ref{lem:one-one}, $\mathcal T^0_{1}$, $\mathcal T^0_{1,2}$, $\mathcal T^0_{1,3}$ and $\mathcal T^0_{2,3}$
are pairwise disjoint. Then, $(U_{q+1}, \cB^0_{\sigma_{6,3}, q+1})$ is a $3$-$\left (q+1, 6, 2(q-4) \right )$ design.

Let $\{u_1,u_2,u_3\}$ be a fixed $3$-subset of $U_{q+1}$.
Define
\begin{align*}
\mathcal T^1= \left \{ \left \{u_1, u_2, u_3, u_4, u_5, u_6\right  \}:
 u_4 \in U_{q+1}\setminus S^0 , u_5   \in U_{q+1} \setminus (S^0 \cup S^1) \right  \},
\end{align*}
where  $S^0=\{u_1, u_2, u_3\} \cup Q(u_1,u_2, u_3)$, $S^1=\left \{ \frac{\sigma_{4,3}+u_i \sigma_{4,2}}{\sigma_{4,2}+u_i \sigma_{4,1}}: 1\le i \le 4 \right \} \bigcup \left \{\sqrt{\frac{\sigma_{4,3}}{\sigma_{4,1}}}  \right \}$,  and $$u_6=\frac{\sigma_{5,3}(u_1,u_2, u_3, u_4, u_5)}{\sigma_{5,2}(u_1,u_2, u_3, u_4, u_5)}.$$
Let $\mathcal T= \mathcal T^0_{1} \cup \mathcal T^1 $.
It is easily checked that $B \in \cB_{\sigma_{6,3},q+1}$  if and only
if $B \in \mathcal T$.
Note that $\# \mathcal T^0_{1}= q-4 $ and $\# \mathcal T^1=\frac{(q+1-\# S^0)(q+1-\# (S^0\cup S^1))}{6}$.
By Lemmas \ref{lem:sigma4-5Sigma63} and  \ref{lem:S11}, $\# (S^0\cup S^1)=11$.
From Lemma \ref{lem:one-one}, $\mathcal T^0_{1}$ and $\mathcal T^1$
are disjoint. Then, $(U_{q+1}, \cB_{\sigma_{6,3}, q+1})$ is a $3$-$\left (q+1, 6, \frac{(q-4)^2}{6} \right )$ design.
This completes the proof.  \\
\rightline{$\square$}

\section{ Infinite families of  BCH codes supporting $t$-designs for $t =3, 4$}
Throughout this section, let $q=2^m$, where  $m$ is a positive integer. 
In this section, we consider the 
narrow-sense BCH code $\C_{(q, q+1, 4, 1)}$ over $\gf(q)$ and its dual, and prove that they are almost MDS, and 
support $4$-designs when $m\ge 5$ is  odd and $3$-designs when  $m\ge 4$ is even. 

 For a positive integer $\ell$, define a $6\times  \ell$ matrix $M_{\ell}$ by
\begin{eqnarray}\label{eq:M}
\left [
\begin{array}{llll} 
u_1^{-3} & u_2^{-3} & \cdots & u_{\ell}^{-3} \\ 
u_1^{-2} & u_2^{-2} & \cdots & u_{\ell}^{-2} \\ 
u_1^{-1} & u_2^{-1} & \cdots & u_{\ell}^{-1} \\ 
u_1^{+1} & u_2^{+1} & \cdots & u_{\ell}^{+1} \\ 
u_1^{+2} & u_2^{+2} & \cdots & u_{\ell}^{+2} \\ 
u_1^{+3} & u_2^{+3} & \cdots & u_{\ell}^{+3} \\ 
\end{array} 
\right ],
\end{eqnarray} 
where $u_1, \cdots,  u_{\ell} \in U_{q+1}$. For $r_1, \cdots, r_i \in \{\pm1, \pm2, \pm3\}$, let $M_{\ell}[r_1, \cdots, r_i]$
denote the submatrix of $M_{\ell}$  obtained by deleting   the rows 
$(u_1^{r_1},  u_2^{r_1},  \cdots ,u_{\ell}^{r_1})$, $\cdots$ ,  $(u_1^{r_i} , u_2^{r_i} , \cdots , u_{\ell}^{r_i})$ of the matrix $M_{\ell}$.

\begin{lemma}\label{lem:sol-rank}
Let $M_{\ell}$ be the matrix given by (\ref{eq:M}) with  $\{u_1, \cdots,  u_{\ell}\} \in  \binom{U_{q+1}}{\ell}$.
Consider the system of homogeneous linear equations defined by
\begin{align}\label{eq:Mx=0}
M_{\ell} (x_1, \cdots,  x_{\ell})^{T}=0.
\end{align}
Then (\ref{eq:Mx=0}) has a nonzero  solution $(x_1, \cdots,  x_{\ell})$ in $\gf(q)^{\ell}$  if and only if
$\rank (M_{\ell})<\ell$, where $\rank (M_{\ell})$ denotes the rank of  the matrix $M_{\ell}$.
\end{lemma}
\begin{proof}
It is obviously that $\rank (M_{\ell})<\ell$ if (\ref{eq:Mx=0}) has a nonzero  solution $(x_1, \cdots,  x_{\ell})$ in $\gf(q)^{\ell}$.

Conversely, suppose that $\rank (M_{\ell})<\ell$. Then, there exists a nonzero vector $\mathbf{x}'=(x_1', \cdots, x_{\ell}') \in \gf(q^2)^{\ell}$
such that $M_{\ell} \mathbf{x}'^{T}=0$. Choose an $i_0 \in \{1, \cdots, \ell \}$  such that $x_{i_0}'\neq 0$.
Put $$\mathbf{x}=(x_1''+x_1''^q, \cdots, x_{i_0}''+x_{i_0}''^q,\cdots, x_{\ell}''+x_{\ell}''^q),$$ 
where $(x_1'', \cdots, x_{\ell}'')= \frac{\alpha}{x_{i_0}'}\mathbf{x'}$ and $\alpha$  is a primitive element of $\gf(q^2)$.
It's easily checked that $M_{\ell} \mathbf{x}^T=0$ and $\mathbf x \in \gf(q)^{\ell} \setminus \{\mathbf 0\}$.
This completes the proof.
\end{proof}

\begin{lemma}\label{lem:rank44}
Let $M_{4}$ be the matrix given by (\ref{eq:M}) with  $\{u_1, u_2, u_3,u_4\} \in  \binom{U_{q+1}}{4}$.
Then   $\rank (M_{4})=4$.
\end{lemma}
\begin{proof}
Assume that $\rank (M_{4})<4$. Then $\det (M_4[2,3])=\frac{\prod_{1\le i<j\le 4} (u_i+u_j)}{\sigma_{4,4}^3} (u_1+u_2+u_3+u_4)$=0,
which is contrary to Lemma \ref{lem:sigma41}. This completes the proof.
\end{proof}

\begin{lemma}\label{lem:rank54}
Let $M_{5}$ be the matrix given by (\ref{eq:M}) with  $\{u_1, \cdots,  u_5\} \in  \binom{U_{q+1}}{5}$.
Then   $\rank (M_{5})=4$ if and only if $\sigma_{5,2}(u_1, \cdots, u_5)=0$.
\end{lemma}
\begin{proof}
First, note that
\begin{eqnarray*}
\left\{ 
\begin{array}{ll}
\det (M_5[3])&= \frac{\prod_{1\le i<j\le 5} (u_i+u_j)}{\sigma_{5,5}^3} \sigma_{5,2},\\
\det (M_5[2])&= \frac{\prod_{1\le i<j\le 5} (u_i+u_j)}{\sigma_{5,5}^3}\left (  \sigma_{5,1} \sigma_{5,2}+ \sigma_{5,5} \sigma_{5,2}^q \right ),\\
\det (M_5[1])&= \frac{\prod_{1\le i<j\le 5} (u_i+u_j)}{\sigma_{5,5}^3}\left (  \sigma_{5,1} \sigma_{5,5}\sigma_{5,2}^q+ \sigma_{5,2}^2  \right ),\\
 \det (M_5[-3])&= \frac{\prod_{1\le i<j\le 5} (u_i+u_j)}{\sigma_{5,5}} \sigma_{5,2}^q,\\
 \det (M_5[-2])&= \frac{\prod_{1\le i<j\le 5} (u_i+u_j)}{\sigma_{5,5}}\left (  \sigma_{5,1}^q \sigma_{5,2}^q+ \sigma_{5,5}^q \sigma_{5,2} \right ),\\
 \det (M_5[-1])& =\frac{\prod_{1\le i<j\le 5} (u_i+u_j)}{\sigma_{5,5}}\left (  \sigma_{5,1}^q \sigma_{5,5}^q\sigma_{5,2}+ \sigma_{5,2}^{2q}  \right ).\\
\end{array}
\right. 
\end{eqnarray*} 
The desired conclusion then follows from Lemma \ref{lem:rank44}. This completes the proof.
\end{proof}

\begin{lemma}\label{lem:rank65}
Let $M_{6}$ be the matrix given by (\ref{eq:M}) with  $\{u_1, \cdots,  u_6\} \in  \binom{U_{q+1}}{6}$.
Then   $\rank (M_{6})<6$ if and only if $\sigma_{6,3}(u_1, \cdots, u_6)=0$.
\end{lemma} 

\begin{proof}
Note that 
$$\det (M_6)= \frac{\prod_{1\le i<j\le 6} (u_i+u_j)}{\sigma_{6,6}^3} \sigma_{6,3},$$
which completes the proof.
\end{proof}

\begin{lemma}\label{lem:sys-wt6}
Let $q=2^m$ with $m$ even and $M_{6}$ be the matrix given by (\ref{eq:M}) with  $\{u_1, \cdots,  u_6\} \in  \binom{U_{q+1}}{6}$.
Let $\{u_1, \cdots, u_6\} \in \cB^1_{\sigma_{6,3},q+1}$
, where $\cB^1_{\sigma_{6,3},q+1}$ is defined by (\ref{eq:B1}).
Then,  the set of all solutions of the system $M_6(x_1, \cdots, x_6)^T=0$ over $\gf(q)^6$ is
$$\left \{ (ax_1, \cdots, a x_6) : a\in \gf(q)  \right \},$$
where $(x_1, \cdots, x_6) $ is  a vector in $\left ( \gf(q)^*\right )^6$.
 \end{lemma}  
 
\begin{proof}
Let $\{u_1, \cdots, u_6\} \in \cB^1_{\sigma_{6,3},q+1}$. 
By Lemma \ref{lem:rank65}, $\rank (M_{6})<6$. By Lemma \ref{lem:sol-rank}, there
exists a nonzero $(x_1,\cdots, x_{6}) \in \gf(q)^6$ such that  $M_6(x_1, \cdots, x_6)^T=0$.
Assume that there is an $i $ ($1\le i\le 6$) such that $x_i=0$. Then the submatrix of the matrix $M_6$
obtained by deleting the $i$-th column has rank less than $5$, which   is contrary to 
Lemma \ref{lem:rank54} and the definition of  $\cB^1_{\sigma_{6,3},q+1}$. 
Thus, for any nonzero solution $(x_1,  \cdots, x_6)\in \gf(q)^6$, we have $x_i\neq 0$, where $1\le i\le 6$.
The desired conclusion then follows. This completes the proof.

\end{proof}

\begin{lemma}\label{lem:wt6-sp}
Let $q=2^m$ with $m$ even and $M_{6}$ be the matrix given by (\ref{eq:M}) with  $\{u_1, \cdots,  u_6\} \in  \binom{U_{q+1}}{6}$.
If there exists a vector $(x_1, \cdots, x_6) \in \left ( \gf(q)^*\right )^6$ such that $M_6(x_1, \cdots, x_6)^T=0$,  then 
 $\{u_1, \cdots, u_6\} \in \cB^1_{\sigma_{6,3},q+1}$
, where $\cB^1_{\sigma_{6,3},q+1}$ is defined by (\ref{eq:B1}).
 \end{lemma}
\begin{proof}
By Lemma \ref{lem:rank65}, $\{u_1, \cdots, u_6\} \in \cB_{\sigma_{6,3},q+1}$.
Assume that $\{u_1, \cdots, u_6\} \in \cB^0_{\sigma_{6,3},q+1}$, without loss of generality, let $\sigma_{5,2}(u_1,  \cdots, u_5)=0$.
By Lemmas \ref{lem:sol-rank} and \ref{lem:rank54}, there exists  a nonzero $(x_1', \cdots, x_5')\in  \gf(q)^5$
such that $M_5(x_1', \cdots, x_5')^T=0$, that is, $M_6(x_1', \cdots, x_5',0)^T=0$.
Note that $$M_6\left (x_1+\frac{x_1}{x_1'} x_1', \cdots, x_5+\frac{x_1}{x_1'} x_5',x_6+\frac{x_1}{x_1'} 0 \right )^T=0.$$
Applying Lemma \ref{lem:rank54}, $\sigma_{5,2}(u_2, \cdots, u_6)=0$, which is contrary to
Lemma \ref{lem:one-one} and $\sigma_{5,2}(u_1,  \cdots, u_5)=0$. This  completes the proof.

\end{proof}

\begin{lemma}\label{lem:trPoly}
Let $f(u)=\tr_{q^2/q} \left ( a u^3+ b u^2+ c u \right )$ where  $(a, b, c) \in \gf(q^2)^3 \setminus  \{\mathbf 0\}$.
Define 
$\mathrm{zero}(f)=\left \{ u\in U_{q+1} : f(u)=0 \right \}$.
Then $\#\left ( \mathrm{zero}(f) \right )\le 6$. Moreover, $\# \left (\mathrm{zero}(f)\right )= 6$
if and only if $a=\frac{\tau}{\sqrt{\sigma_{6,6}}}$, $b= \frac{\tau\sigma_{6,1}}{\sqrt{\sigma_{6,6}}} $ and $c=\frac{\tau\sigma_{6,2}}{\sqrt{\sigma_{6,6}}}$, 
where $\{u_1, \cdots, u_6\} \in \cB_{\sigma_{6,3},q+1}$ and $\tau \in \gf(q)^*$.
\end{lemma}
\begin{proof}
When $u\in U_{q+1}$, one has
\begin{align}\label{eq:f-c-U}
f(u)=\frac{1}{u^3}\left ( a u^6 + b u^5 +cu^4 + c^q u^2 +b^q u+a^q   \right ).
\end{align}
Thus, $\#\left ( \mathrm{zero}(f) \right )\le 6$.

Assume that $\#\left ( \mathrm{zero}(f) \right )= 6$. From  (\ref{eq:f-c-U}), 
there exists $\{u_1, \cdots, u_6\} \in U_{q+1}$ 
such that $f(u)=\frac{a \prod_{i=1}^6 (u+u_i)}{u^3}$.
By Vieta's formula, $b=a\sigma_{6,1}$, $c=a\sigma_{6,2}$,$0=\sigma_{6,3}$, $c^q=a\sigma_{6,6}\sigma_{6,2}^q$,
$b^q=a\sigma_{6,6}\sigma_{6,1}^q$ and $a^q=a\sigma_{6,6}$.
One obtains $a= \frac{\tau}{\sqrt{\sigma_{6,6}}}$ from $a^{q-1}= \sigma_{6,6}$, where $\tau \in \gf(q)^*$. 
Then, $b=\frac{\tau\sigma_{6,1}}{\sqrt{\sigma_{6,6}}}$ and  $c=\frac{\tau\sigma_{6,2}}{\sqrt{\sigma_{6,6}}}$.

Conversely, assume that $a=\frac{\tau}{\sqrt{\sigma_{6,6}}}$, $b= \frac{\tau\sigma_{6,1}}{\sqrt{\sigma_{6,6}}} $ and $c=\frac{\tau\sigma_{6,2}}{\sqrt{\sigma_{6,6}}}$, 
where $\{u_1, \cdots, u_6\} \in \cB_{\sigma_{6,3},q+1}$ and $\tau \in \gf(q)^*$. Then
$f(u)=\frac{a \prod_{i=1}^6 (u+u_i)}{u^3}$. Thus, $\mathrm{zero}(f)=\{u_1,  \cdots, u_6\}$ and $\#(\mathrm{zero}(f))=6$.
\end{proof}

\subsection{A class of narrow-sense  BCH codes with length $2^m+1$}

 We are now ready to prove the following result about the code $\C_{(q, q+1, 4,1)}$.

\begin{theorem}\label{thm:C-(q,q+1,4,1)}  
Let $q=2^m$ with $m\ge 4$ being an  integer. Then the narrow-sense BCH code $\C_{(q, q+1, 4,1)}$ over $\gf(q)$ 
has parameters $[q+1, q-5, d]$, where $d=6$ if $m$ is odd and $d=5$ if $m$ is even. 
\end{theorem} 

\begin{proof}
Put $n=q+1$. 
Let $\alpha$ be a generator of $\gf(q^2)^*$ and $\beta=\alpha^{q-1}$. Then $\beta$ is a primitive $n$-th root of unity 
in $\gf(q^2)$, that is, $\beta$ is a generator of the cyclic group $ \in U_{q+1}$. Let $g_i(x)$ denote the minimal polynomial of  $\beta^i$ over $\gf(q)$, 
where $i\in \{1,2,3\}$. 
Note that $g_i(x)$ has only the roots $\beta^i$ and $\beta^{-i}$. 
One deduces that $g_1(x)$, $g_2(x)$ and  $g_3(x)$ are pairwise distinct irreducible polynomials of degree $2$. 
By definition, $g(x):=g_1(x)g_2(x)g_3(x)$ is the generator polynomial of $\C_{(q, q+1, 4,1)}$. Therefore, 
the dimension of $\C_{(q, q+1, 4,1)}$ is $q+1-6$. 
Note that 
$g(x)$ has only the roots $\beta^{-3}, \beta^{-2}, \beta^{-1}, \beta, \beta^2$ and $\beta^3$. By the BCH bound, the minimum weight 
of $\C_{(q, q+1, 4,1)}$ is at least $4$. 
Put $\gamma=\beta^{-1}$.  Then $\gamma^{q+1}=\beta^{-(q+1)}=1$. 
It then follows from Delsarte's theorem that the trace expression of $\C_{(q, q+1, 4,1)}^\perp$ is given by 
\begin{eqnarray}\label{eq:dualC-Tr}
\C_{(q, q+1, 4,1)}^\perp=\{\bc_{(a,b,c)}: a, b, c \in \gf(q^2)\}, 
\end{eqnarray} 
where $\bc_{(a,b,c)}=(\tr_{q^2/q}(a\gamma^i +b \gamma^{2i} +c \gamma^{3i}))_{i=0}^q$. 

Define 
\begin{eqnarray}
H=\left[ 
\begin{array}{rrrrrr}
1  &  \gamma^{-3} & \gamma^{-6} & \gamma^{-9} & \cdots & \gamma^{-3q}\\
1  & \gamma^{-2} & \gamma^{-4} &  \gamma^{-6} & \cdots & \gamma^{-2q} \\
1  & \gamma^{-1} & \gamma^{-2} & \gamma^{-3} & \cdots & \gamma^{-q} \\
1  & \gamma^{+1} & \gamma^{+2} & \gamma^{+3} & \cdots & \gamma^{+q} \\
1  & \gamma^{+2} & \gamma^{+4} &  \gamma^{+6} & \cdots & \gamma^{+2q} \\
1  &  \gamma^{+3} & \gamma^{+6} & \gamma^{+9} & \cdots & \gamma^{+3q}
\end{array}
\right].  
\end{eqnarray} 
It is easily seen that $H$ is a parity-check matrix of $\C_{(q, q+1, 4,1)}$, i.e., 
\begin{align}\label{eq:C-H}
\C_{(q, q+1, 4,1)}=\{\bc \in \gf(q)^{q+1}: \bc H^T=\bzero\}. 
\end{align}

Let $m$ be odd. Note that $d\ge 4$.
Assume that $d=4$. Then
there exist $\{u_1, \cdots, u_4\} \in \binom{U_{q+1}}{4}$ and $(x_1, \cdots, x_{4})\in  \left (\gf(q)^* \right )^4$
such that $M_4(x_1, \cdots, x_{4})^T=0$. Thus $\rank (M_4) <4$, which is contrary to Lemma \ref{lem:rank44}.
Assume that $d=5$. Then
there exist $\{u_1, \cdots, u_5\} \in \binom{U_{q+1}}{5}$ and $(x_1, \cdots, x_{5})\in  \left (\gf(q)^* \right )^5$
such that $M_5(x_1, \cdots, x_{5})^T=0$. By Lemma \ref{lem:rank54},  $\rank (M_5) <5$ and $\sigma_{5,2}=0$, which is contrary to Lemma \ref{lem:B52odd}.
Thus, $d\ge 6$. By Theorem \ref{thm:esp4-design}, $ \cB_{\sigma_{6,3},q+1}\neq \emptyset$.
Choose $\{u_1, \cdots, u_6\} \in \cB_{\sigma_{6,3},q+1}$. By Lemma \ref{lem:sol-rank},
there exists $(x_1, \cdots, x_6) \in \left ( \gf(q)^* \right )^6$ such that $M_6 (x_1, \cdots, x_6)^T=0$.
Set $\bc=(c_1, \cdots, c_{q+1})$ where
\begin{eqnarray}
c_i=\left\{ 
\begin{array}{rr}
x_{j},  &  \text{ if } i=i_j,\\
0,  & \text{ otherwise}, \\
\end{array}
\right.  
\end{eqnarray} 
where $\gamma_{i_j}$ is given by $u_j= \gamma^{i_j}$ ($j \in \{1, \cdots, 6\}$).
By (\ref{eq:C-H}), $\bc \in \C_{(q, q+1, 4,1)}$ and $\wt(\bc)=6$. Thus, $d=6$.

The proof for the case  $m$ even is similar as the case $m$ odd. And the detail is omitted. 
This completes the proof. 
\end{proof}

\begin{theorem}\label{thm:dualC} 
Let $q=2^m$ with $m\ge 4$ and $\C_{(q, q+1, 4,1)}^\perp$ be the  dual of  the narrow-sense BCH code $\C_{(q, q+1, 4,1)}$ over $\gf(q)$.
Then $\C_{(q, q+1, 4,1)}^\perp$  has parameters $[q+1, 6, q-5]$. In particular,  $\C_{(q, q+1, 4,1)}$ is a near  MDS code if $m$ is odd.
\end{theorem} 

\begin{proof}
From Theorems \ref{thm:esp4-design}  and \ref{thm:esp3-design}, $\cB_{\sigma_{6,3},q+1} \neq \emptyset$.
The desired conclusion then follows from Lemma \ref{lem:trPoly} and Equation (\ref{eq:dualC-Tr}).
This completes the proof.
\end{proof}

\subsection{An infinite class of near MDS codes supporting $4$-designs}

\begin{theorem}\label{thm:code-design-esp-odd} 
Let $q=2^m$ with $m \geq 5$ odd. Then, the incidence structure $$\left ( \cP\left (\C_{(q,q+1,4,1)} \right ), \cB_6\left (\C_{(q,q+1,4,1)} \right )  \right )$$
from the minimum weight codewords in $\C_{(q, q+1, 4,1)}$ is isomorphic to  $(U_{q+1}, \cB_{\sigma_{6,3}, q+1})$ .
\end{theorem}
\begin{proof}
Using Lemma \ref{lem:rank65}, the desired conclusion then follows by a similar  discussion as in the proof of Theorem \ref{thm:C-(q,q+1,4,1)}. This completes the proof.
\end{proof}

The theorem below makes a breakthrough in 71 years in the sense that it presents the first family of linear codes 
 supporting an infinite family of $4$-designs since the first linear code holding a $4$-design was discovered 71 years ago 
by Golay \cite{Golay49}.

\begin{theorem}\label{thm:odd-4designs} 
Let $q=2^m$ with $m \geq 5$ odd. Then, 
the minimum weight codewords in $\C_{(q, q+1, 4,1)}$ support a $4$-$(2^{m}+1, 6, 2^{m-1}-4)$ design
and 
the minimum weight codewords in $\C_{(q, q+1, 4,1)}^\perp$ support a $4$-$(q+1, q-5, \lambda)$ design 
with 
$$ 
\lambda=\frac{q-8}{30} \binom{q-5}{4}. 
$$ 
\end{theorem}

\begin{proof}
The desired conclusion follows from Theorems \ref{thm:code-design-esp-odd}, \ref{thm:esp4-design} and \ref{thm-121FW}.
This completes the proof.
\end{proof}

\begin{example}\label{exam-4design26} 
Let $q=2^5$. Then  $\C_{(q, q+1, 4,1)}$ has parameters $[33,27,6]$. The dual  $\C_{(q, q+1, 4,1)}^\perp$ has parameters 
$[33,6,27]$ and weight distribution 
\begin{eqnarray*}
1 + 1014816z^{27} +  1268520z^{28} +  20296320z^{29} +  64609952z^{30} + \\ 
210132384z^{31} + 399584823 z^{32} + 376835008 z^{33}. 
\end{eqnarray*} 
The codewords of weight $6$ in $\C_{(q, q+1, 4,1)}$ supports a $4$-$(33,6,12)$ design, and the codewords of weight 
$27$ in $\C_{(q, q+1, 4,1)}^\perp$ support a $4$-$(33,27,14040)$ design. 
\end{example}  

In Example \ref{exam-4design26}, the code $\C_{(q, q+1, 4,1)}$ has a codeword of weight $i$ for all $i$ with $6 \leq i \leq 33$. 
Hence, the Assmus-Mattson Theorem cannot prove that the codes  in Theorem \ref{thm:odd-4designs} support $4$-designs. 
It is open if the generalised Assmus-Mattson theorem in \cite{TDX19} can prove that the codes  in Theorem \ref{thm:odd-4designs} support $4$-designs. 
It looks impossible to prove that the codes in Theorem \ref{thm:odd-4designs} support $4$-designs with the automorphism groups of the codes due to the following: 
\begin{enumerate}
\item Except the Mathieu groups M11, M12, M23, M24, the alternating group $A_n$ and the symmetric group $S_n$, no 
finite permutation groups are more than $3$-transitive \cite{BJL}. 
\item No infinite family of $4$-homogeneous permutation groups is known. 
\end{enumerate} 
It would be a very interesting problem to determine the automorphism groups of the codes in Theorem \ref{thm:odd-4designs}.

\subsection{An infinite class of linear codes supporting Steiner systems $S(3,5,4^m+1)$}

\begin{theorem}\label{thm:code-design-esp-even} 
Let $q=2^m$ with $m \geq 4$ even. Then, the incidence structure $$\left ( \cP\left (\C_{(q,q+1,4,1)} \right ), \cB_5\left (\C_{(q,q+1,4,1)} \right )  \right )$$
from the minimum weight codewords in $\C_{(q, q+1, 4,1)}$ is isomorphic to  $(U_{q+1}, \cB_{\sigma_{5,2}, q+1})$,
and   the incidence structure $$\left ( \cP\left (\C_{(q,q+1,4,1)} \right ), \cB_6 \left (\C_{(q,q+1,4,1)} \right )  \right )$$ is isomorphic to  $(U_{q+1}, \cB^{1}_{\sigma_{6,3}, q+1})$.
Moreover, the incidence structure $$\left ( \cP \left (\C_{(q,q+1,4,1)}^{\perp} \right ), \cB_{q-5} \left (\C_{(q,q+1,4,1)}^{\perp} \right )  \right )$$ is isomorphic to  
the complementary  incidence structure of $(U_{q+1}, \cB_{\sigma_{6,3}, q+1})$
\end{theorem}
\begin{proof}
Using Lemma \ref{lem:rank54}, by a similar discussion as as in the proof of Theorem \ref{thm:C-(q,q+1,4,1)},
we can prove that the incidence structure $$\left ( \cP\left (\C_{(q,q+1,4,1)} \right ), \cB_5\left (\C_{(q,q+1,4,1)} \right )  \right )$$
isomorphic to  $(U_{q+1}, \cB_{\sigma_{5,2}, q+1})$. Employing  Lemma \ref{lem:wt6-sp}, 
we can prove that $$\left ( \cP\left (\C_{(q,q+1,4,1)} \right ), \cB_6 \left (\C_{(q,q+1,4,1)} \right )  \right )$$ is isomorphic to  $(U_{q+1}, \cB^{1}_{\sigma_{6,3}, q+1})$.
The last statement then follows from Equation (\ref{eq:dualC-Tr}) and Lemma \ref{lem:trPoly}. This completes the proof.
\end{proof}

\begin{theorem}\label{thm:even-3designs} 
Let $q=2^m$ with $m \geq 4$ even. Then, 
the minimum weight codewords in $\C_{(q, q+1, 4,1)}$ support a $3$-$(2^{m}+1, 5, 1)$ design, 
i.e., 
a Steiner system $S(3,5,2^m+1)$,
and 
the minimum weight codewords in $\C_{(q, q+1, 4,1)}^\perp$ support a $3$-$(q+1, q-5, \lambda)$ design 
with 
$$ 
\lambda=\frac{(q-4)^2}{120} \binom{q-5}{3}. 
$$ 
Furthermore, the codewords of  weight $6$  in $\C_{(q, q+1, 4,1)}$ support a $3$-$\left (q+1, 6, \frac{(q-4)(q-16)}{6} \right )$ design if $m \geq 6$.
\end{theorem} 

\begin{proof}
The desired conclusion follows from Theorems \ref{thm:code-design-esp-even}, \ref{thm:espsteiner}, \ref{thm:esp3-design} and Corollary \ref{cor:B1-WT6}.
This completes the proof.
\end{proof}

There are two different constructions of an infinite family of Steiner systems $S(3, q+1, q^m+1)$ for $q$ being a prime 
power and $m \geq 2$. 
The first produces the spherical designs due to \citet{Witt38}, which is based on the action of $\PGL_2( \gf(q^m))$ on the base block 
$\gf(q) \cup \{\infty\}$. The automorphism group of the spherical design contains the group 
$\PGaL_2(\gf(q^m))$. The second construction was proposed in \cite{KeyWagner86}, and is based on affine spaces. 
The Steiner systems  $S(3, q+1, q^m+1)$ from the two constructions are not isomorphic \cite{KeyWagner86}.

When $m \in \{2,3\}$, the Steiner  system $S(3,5,4^{m}+1)$ of  Theorem \ref{thm:even-3designs} 
is isomorphic to the spherical design with the same parameters. We conjecture that they are isomorphic 
in general, but do not have a proof. The contribution  of Theorem \ref{thm:even-3designs} is a coding-theoretic 
construction of the spherical systems  $S(3,5,4^m+1)$. 

\begin{example} 
Let $q=2^4$. Then  $\C_{(q, q+1, 4,1)}$ has parameters $[17,11,5]$ and weight distribution 
\begin{eqnarray*}
1 + 1020z^5 + 224400 z^7 +  3730650z^8 + 55370700z^9 +  669519840z^{10} +  \\ 
 6378704640z^{11} +  47857084200z^{12}  +  
 276083558100z^{13}  + 1183224112800z^{14} +  \\ 
 3549668972400z^{15} + 
  6655630071165z^{16} +  5872614694500z^{17}. 
\end{eqnarray*} 
The codewords of weight $5$ in $\C_{(q, q+1, 4,1)}$ support a Steiner system $S(3,5,17)$. 
 
The dual $\C_{(q, q+1, 4,1)}^\perp$ has parameters $[17,6,11]$ and weight distribution 
\begin{eqnarray*}
1+  12240z^{11} + 35700 z^{12} + 244800z^{13} + 1203600 z^{14} + 3292560z^{15} + 
 6398715 z^{16} + 5589600 z^{17}.  
\end{eqnarray*} 
The codewords of weight $11$ in $\C_{(q, q+1, 4,1)}^\perp$ support a $3$-$(17,11,198)$ design. 
\end{example} 

This example shows that the Assmus-Mattson Theorem cannot prove that the codes $\C_{(q, q+1, 4,1)}$  
and $\C_{(q, q+1, 4,1)}^\perp$ support $3$-designs. 
It is open if the generalised Assmus-Mattson theorem in \cite{TDX19} can prove that the codes  in Theorem \ref{thm:even-3designs} support $4$-designs. 
It is also open if the automorphism groups of the codes 
can prove that the codes support $3$-designs. 

\section{Summary and concluding remarks}

This paper settled the 71-year-old open problem by presenting an infinite family of near MDS codes of length $2^{2m+1}+1$ over 
$\gf(2^{2m+1})$ holding an infinite family of $4$-$(2^{2m+1}+1, 6, 2^{2m}-4)$ designs. Hence, these codes have nice applications in combinatorics.  
It would be nice if the automorphism groups of the linear codes could be determined.    

An interesting open problem is whether there exists an infinite family of linear codes holding an infinite family of $t$-designs  
for $t \ge 5$. Another open problem is whether there is a specific linear code supporting a nontrivial $6$-design.



\end{document}